\documentclass{llncs}

\newcommand*{\N}{\ensuremath{\mathbb{N}}}
\newcommand*{\nat}{\ensuremath{\mathbb{N}}}

\newcommand{\To}{\longrightarrow}

\newcommand{\manyk}[1]{{#1}_1,        \ldots       ,       {#1}_{k}}

\newcommand{\manyr}[1]{{#1}_1,        \ldots       ,       {#1}_{r}}

\newcommand{\manym}[1]{{#1}_1,        \ldots       ,       {#1}_{m}}
\newcommand{\many}[2]{{#1}_1,        \ldots       ,       {#1}_{#2}}

\renewcommand{\k}{{\sf K}}
\renewcommand{\d}{{\sf D}}
\renewcommand{\t}{{\sf T}}
\newcommand{\kf}{{\sf K4}}
\newcommand{\df}{{\sf D4}}
\newcommand{\sr}{{\sf S4}}

\renewcommand{\j}{{\sf J}}
\newcommand{\jd}{{\sf JD}}
\newcommand{\jdf}{{\sf JD4}}
\newcommand{\jf}{{\sf J4}}
\newcommand{\jt}{{\sf JT}}
\newcommand{\lp}{{\sf LP}}

\usepackage{complexity}

\newcommand{\term}[2]{#1\!  :_{#2} \!}
\newcommand{\sterm}[1]{#1  : \!}

\usepackage{mathpartir}
\newcommand{\cs}{\mathcal{CS}}

\renewcommand{\M}{\mathcal{M}}
\newcommand{\T}{\mathcal{T}}
\newcommand{\F}{\mathcal{F}}
\newcommand{\I}{\mathcal{I}}
\newcommand{\V}{\mathcal{V}}

\newcommand{\true}{\text{\it true}}
\newcommand{\false}{\text{\it false}}

\usepackage{amsmath}
\usepackage{amsfonts}
\usepackage{amssymb}

\renewcommand{\A}{{\mathcal{E}}}        
\renewcommand{\E}{{\mathcal{E}}}

\newcommand{\ver}{\hookrightarrow}
\newcommand{\inver}{\hookleftarrow}

\binoppenalty=\maxdimen
\relpenalty=\maxdimen

\title{\NEXP-completeness and Universal Hardness Results for Justification Logic}
\author{Antonis Achilleos}
\institute{
	The Graduate Center of CUNY \\
	365 Fifth Avenue\\
	New York, 
	NY 10016 
	USA	
	\\
		 \email{
	aachilleos@gc.cuny.edu
}
	 	}

\begin{document}
	
	\maketitle
	
	\begin{abstract} 
	We provide a lower complexity bound for the satisfiability problem of a multi-agent justification logic, establishing that the general \NEXP\ upper bound from our previous work is tight. We then use a simple modification of the corresponding reduction to prove that satisfiability for all multi-agent justification logics from there 
%	with an axiomatically appropriate and schematic constant specification 
	is $\Sigma_2^p$-hard -- given certain reasonable conditions. Our methods improve on these required conditions for the same lower bound for the single-agent justification logics, proven by Buss and Kuznets in 2009, thus answering one of their open questions.
	\end{abstract}

\section{Introduction}
Justification Logic is the logic of justifications. Where in Modal Epistemic Logic we use formulas of the form $\Box \phi$ to denote that $\phi$ is known (or believed, etc), in Justification Logic, we use $\sterm{t}\phi$ to denote that $\phi$ is known \emph{for reason $t$} (i.e. $t$ is a \emph{justification} for $\phi$). 
Artemov introduced \lp, the first justification logic, in 1995 \cite{Art01BSL}, originally as a link between Intuitionistic Logic and Peano Arithmetic. Since then the field has expanded significantly, both in the variety of logical systems 
%in the family 
and in the fields it interacts with and is applied to (see \cite{Art2008LAI,Art08RSL} for an overview). 

In \cite{multi2} Yavorskaya introduced two-agent \lp\ with agents whose justifications may interact. We  studied the complexity of a generalization in \cite{Achilleos2014CLIMA} and \cite{Achilleos2014EUMAS}, discovering that unlike the case with single-agent Justification Logic 
as studied in
\cite{Kuz00CSL,Kuz08PhD,Kuz09LC,newlower,Achilleos2014JCSS}, the complexity of satisfiability jumps to \PSPACE- and \EXP-completeness when two or three agents are involved respectively, given appropriate interactions.
% in the system. 
 In fact, the upper bound we proved was that all logics in this family have their satisfiability problem in \NEXP\ -- under reasonable assumptions.
%under an axiomatically appropriate, schematic constant specification, i.e. when all axiom \emph{schemes} are justified by some justification constant (and all constants justify schemes). 

The \NEXP\ upper complexity bound was not met with the introduction of a \NEXP-hard logic in \cite{Achilleos2014EUMAS}. The main contribution of this paper is that we present a \NEXP-hard justification logic from the family that was introduced in \cite{Achilleos2014EUMAS}, thus establishing that the general upper bound is tight. 

In general, the complexity of the satisfiability problem for a justification logic tends to be lower than the complexity of its corresponding modal logic\footnote{That is, the modal logic that is the result of substituting all justification terms in the axioms with boxes and adding the Necessitation rule.} (given the usual complexity-theoretic assumptions).
For example, while satisfiability for \k, \d, \kf, \df, \t, and \sr\ is \PSPACE-complete, the complexity of the corresponding justification logics (\j, \jd, \jf, \jdf, \jt, and \lp\ respectively) is in the second level of the polynomial hierarchy (in $\Sigma_2^p$, specifically). In the multi-agent setting we have already examined, this is still the case: many justification logics that so far demonstrate a complexity jump (to \PSPACE- or \EXP-completeness) have corresponding modal logics with an \EXP-complete satisfiability problem (c.f. \cite{Spaan93PhD,DBLP:conf/tableaux/Demri00,Achilleos2014JCSS,Achilleos14DiamondFreeArxiv}).
It is notable that, assuming $\EXP \neq \NEXP$, this is the first time we have a justification logic with a higher complexity than its corresponding modal logic, and, in fact, the reduction we use makes heavy use of the effects of the way a justification term is constructed.

In a justification logic, the logic's axioms are justified by \emph{constants}, 
%which are 
a kind of minimal (not analyzable) justification.
A constant specification is part of the description of a justification logic and specifies exactly which constants justify which axioms.
There are certain standard assumptions  we often need to make when studying the complexity of a justification logic. 
One is that the logic has an axiomatically appropriate constant specification, which means that all axioms of the logic are justified by at least one justification constant.
Another is that the logic has a schematic constant specification, which means that each constant justifies a certain number of axiom \emph{schemes} (perhaps none) and nothing else. Finally, the third assumption is that the constant specification is schematically injective, that is, it is schematic and each constant justifies at most one scheme.

It is known that for (single-agent) justification logics \j, \jt, \jf, and \lp, the satisfiability problem is in $\Sigma_2^p$ for a schematic constant specification (\cite{Kuz00CSL})  and for  \jd, \jdf, the satisfiability problem is in $\Sigma_2^p$ for an axiomatically appropriate and schematic constant specification (\cite{Kuz08PhD,Achilleos2014JCSS}). As for the lower bounds, Milnikel has proven (\cite{Mil07APAL}) that \jf-satisfiability is $\Sigma_2^p$-hard for an axiomatically appropriate and schematic constant specification and that \lp-satisfiability is $\Sigma_2^p$-hard for an axiomatically appropriate, (schematic,) and schematically injective constant specification. Following that, Buss and Kuznets gave a general lower bound in \cite{newlower}, proving that for all the above logics, satisfiability is $\Sigma_2^p$-hard for an axiomatically appropriate, (schematic,) and schematically injective constant specification. This raised the question of whether the condition that the constant specification is schematically injective is a necessary one, which is answered in this paper.\footnote{The answer is `no'.}

We present a general lower bound, which applies to all logics from \cite{Achilleos2014EUMAS}. This includes all the single-agent logics whose complexity was studied in \cite{Kuz00CSL,Kuz08PhD,newlower,Achilleos2014JCSS}. 
In fact, 
Buss and Kuznets gave the same general lower bound for all the single-agent cases in \cite{newlower} and it is reasonable to expect that we could simply apply their techniques and achieve the same result in this general multi-agent setting. Our method, however, presents the following two advantages: it is a relatively simple reduction, a direct simplification of the more involved \NEXP-hardness reduction and very similar to Milnikel's method from \cite{Mil07APAL}; it is also an improvement of their result, even if it does not improve the bound itself in that for our results the requirements are that the constant specification is axiomatically appropriate and schematic -- and not that it is schematically injective as well.
%
%The improvement is in the assumptions we need to make about the constant specification. 
%In particular, for most upper bounds we need to assume the logic is equipped with an axiomatically appropriate, schematic constant specification and this matches our assumptions for the lower bound. On the other hand, Buss and Kuznets require the additional assumption that the constant specification is schematically injective, i.e. that a constant can justify at most one scheme. 
In particular this means that we provide for the first time a tight lower bound for the full \lp\ (\lp\ where all axioms are justified by all constants). The disadvantage of our method is that, unlike the one of Buss and Kuznets, it cannot be adjusted to work on the reflected fragments of justification logic, the fragment which includes only the formulas of the form $\sterm{t}\phi$.
%These are reasonable assumptions: as Buss and Kuznets showed in their paper, 

\section{Background}
\label{sec:defs}
We present the family of multiagent justification logics from \cite{Achilleos2014EUMAS}, its semantics and $*$-calculus, and notation we will be using. The definitions and propositions in this section can be found in \cite{Achilleos2014CLIMA,Achilleos2014EUMAS}.

\subsection{Syntax and Axioms}

%In this paper, if $n\in \nat$, $N$ will be the set $\{ 1, 2, \ldots , n \}$.
The justification terms of the language $L_n$  include constants $c_1, c_2, c_3, \ldots$ and variables $x_1, x_2, x_3, \ldots$ and $t::= x\mid c\mid [t + t]\mid [t\cdot t]\mid ! t$. The set of terms 
%will be referred to as 
is called
$Tm$. 
%		We also use a set $Pvar$ of propositional variables. 
The $n$ agents are represented by the positive integers $i \in N=\{1,\ldots,n\}$. 
%such that $1\leq i \leq n$.
The propositional variables will usually (but not always, as will be evident in the following section) be $p_1,p_2,\ldots$.
Formulas of the language $L_n$ are defined: $\phi ::= \bot \mid p\mid \neg \phi \mid \phi \rightarrow \phi \mid \phi \wedge \phi \mid \phi \vee \phi \mid \term{t}{i}\phi $, but depending on convenience we may treat some connectives as constructed from others. We are particularly interested in 
%the following fragment of the language: 
$rL_n = \{ \term{t}{i}\phi \in L_n \}$.
%include all propositional variables and if $\phi, \psi$ are formulas, $i \in N$ (i.e. $i$ is an agent) and $t$ is a term, then the following are also formulas of $L_n$: $\bot, \phi \rightarrow\psi , \term{t}{i} \phi $.  The remaining propositional connectives, whenever needed, are treated as constructed from $\rightarrow$ and $\bot$ in the usual way: $\neg a := a \rightarrow \bot$, $a\vee b := \neg a \rightarrow b$, and $a\wedge b := \neg (\neg a \vee \neg b)$.
%The operators $\cdot, +$ and $!$ are explained by the following axioms. 
Intuitively, $\cdot$ applies a justification for a statement $A \rightarrow B$ to a justification for $A$ and gives a justification for $B$. Using $+$ we can combine two justifications and have a justification for anything that can be justified by any of the two initial terms -- much like the concatenation of two proofs. 
Finally, $!$ is a unary operator called the proof checker. Given a justification $t$ for $\phi$, $!t$ justifies the fact that $t$ is a justification for $\phi$.

If $\subset, \ver$ are binary relations on the agent set $\{1,\ldots ,n \}$ and for every agent $i$, $F(i)$ is a (single-agent) justification logic (we assume  $F(i)\in \{\j,\jd,\jt\}$), then  justification logic  $J = (n,\subset,\ver,F)_{\cs}$ has the axioms as seen on Table \ref{table:axioms} and modus ponens. The binary relations $\subset, \ver$ determine the interactions among the agents: $\subset$ determines the instances of the Conversion axiom, while $ \ver$ the instances of the Verification axiom, so if $i \subset i$, then the justifications of agent $j$ are also valid justifications for agent $i$ (i.e. we have axiom $\term{t}{j}\phi \rightarrow  \term{t}{i}\phi$), while if $i \ver i$, then the justifications of agent $j$ can be \emph{verified} by agent $i$ (i.e. we have axiom $\term{t}{j}\phi \rightarrow  \term{!t}{i}\term{t}{j}\phi$). 
$F$ assigns a single-agent justification logic to each agent. 
We would assume $F(i)$ is one of \j, \jd, \jt, \jf, \jdf, and \lp,
% since these are the single-agent logics 
but since Positive introspection is a special case of Verification, we can limit the choices for $F(i)$ to logics without Positive Introspection (i.e. \j, \jd, and \jt).
$\cs$ is called a constant specification. It introduces justifications for the axioms and is explained in Table \ref{table:axioms} together with the axioms. 
We also define $i \supset j$ iff $j \subset i$ and $i \inver j$ iff $j \ver i$.

\begin{table}[t]
	General axioms (for every agent $i$):
\begin{description}
	\item[Propositional Axioms:] Finitely many schemes of classical propositional logic;
	\item[Application:] $\term{s}{i}(\phi\rightarrow \psi) \rightarrow (\term{t}{i}\phi \rightarrow \term{[s\cdot t]}{i} \psi)$;
	\item[Concatenation:] 
	$\term{s}{i}\phi \rightarrow \term{[s + t]}{i} \phi$, 
	$\term{s}{i}\phi \rightarrow \term{[t + s]}{i} \phi$.
\end{description}
	Agent-dependent axioms (depending on $F(i)$):
\begin{description}
	\item[Factivity:] for every agent $i$, such that $F(i) 
	=
%	\in \{
	\jt
%	,\lp\}
	$, 
	$\term{t}{i}\phi \rightarrow \phi$;
	\item[Consistency:] for every agent $i$, such that $F(i) 
	=
%	\in \{
	\jd
%	,\jdf\}
	$, $\term{t}{i}\bot \rightarrow \bot$.
%	\item[Positive Introspection:] for every $i$, such that $F(i) \in \{\jf,\lp\}$, $\term{t}{i}\phi \rightarrow  \term{! t}{i} \term{t}{i}\phi$;
\end{description}
Interaction axioms (depending on the binary relations $\subset$ and $\ver$): % and they are the interaction axioms.
\begin{description}
	\item[Conversion:] for every $i \supset j$, $\term{t}{i}\phi \rightarrow  \term{t}{j}\phi$;
	\item[Verification:] for every $i \inver j$, $\term{t}{i}\phi \rightarrow  \term{! t}{j} \term{t}{i}\phi$.
\end{description}
%To complete the description of  justification logic $(n,\subset,\ver,F)_{\cs}$, a \emph{constant specification} $\cs$ is needed:
A constant specification for $(n,\subset,\ver,F)$ is any set of formulas of the form $\term{c}{i}A$, where $c$ a justification constant, $i$ an agent, and $A$ an axiom of the logic from the ones above.
We say that axiom $A$ is justified by a constant $c$ for agent $i$ when $\term{c}{i}A \in \cs$.
%Then we can introduce our final axiom,
\begin{description}
	\item[Axiom Necessitation (AN):] $\term{t}{i}\phi$, where 
	either 
	$\term{t}{i}\phi \in \cs$ or $t=!s$ and $\phi = \term{s}{j}\psi$ an instance of Axiom Necessitation.
\end{description}
\label{table:axioms}
\caption{The axioms of $(n,\subset,\ver,F)_{\cs}$}
\end{table}
%Axiom Necessitation will be called AN for short.
In this paper we will be making the assumption that the constant specifications are 
\emph{axiomatically appropriate:}
each axiom is justified by at least one constant; and
\emph{schematic:} 
every constant justifies only a certain number (0 or more) of the logic's axiom schemes 
(Table \ref{table:axioms}) 
%from the ones above 
-- as a result, every constant justifies a finite number of axiom schemes, but either 0 or infinite axioms, while if $c$ justifies $A$ for $i$ and $B$ results from $A$ and substitution, then $c$ justifies $B$ for $i$. 

\paragraph{We use the following notation and conventions:}
For justification terms $\manyk{t}$ and formulas $\manyk{\phi}$, term  $[t_1+t_2+\cdots + t_k]$ is defined as $ [[t_1 + t_2 + \cdots + t_{k-1}] + t_k ]$,  $[t_1 \cdot t_2  \cdots t_k]$ is defined as  $ [[t_1 \cdot t_2  \cdots  t_{k-1}] \cdot t_k ]$, and $(\phi_1\wedge \phi_2\wedge\cdots \wedge \phi_k)$  as $ ((\phi_1 \wedge \phi_2 \wedge \cdots \wedge \phi_{k-1}) \wedge \phi_k )$ when $k>2$. 
We often identify conjunctions of formulas with sets of such formulas, as long as these can be used interchangeably.
For set of indexes $A$ and $\Phi = \{\term{t_a}{i_a}\phi_a \mid a \in A  \}$, we define	 $\Phi^{\#_i} = \{\phi_a \mid a \in A, i_a = i \}$ and $*\Phi = \{*_{i_a}(t_a,\phi_a) \mid a \in A  \}$.
%, while $\{\Phi_a \mid a \in A  \}$
Often we identify $0,1$ with $\bot, \top$ respectively, as long as it is not a source of confusion. 
%For a non-negative integer $x\in \N$, let $bin(x) = bin_0(g),\ldots , bin_{\log g}(g)$ be its binary representation.
%Finally, for every propositional or first-order formula $\psi$ we introduce propositional variables $[\psi]^\top$ and $[\psi]^\bot$. 
%
%As Proposition  \ref{internalization} demonstrates,
%Modal Logic's full Necessitation Rule is a derived property of Justification Logic:
%
%
%
\begin{lemma}[Internalization Property, \cite{Achilleos2014EUMAS}, but originally \cite{Art01BSL}]\label{internalization}
	For an axiomatically appropriate constant specification $\cs$,  %with respect to $I \subseteq N$, 
	if $\vdash \phi$, then for any $i \in N$ there is some term $t$ such that for any $\phi'$ substitution instance of $\phi$,  $\vdash \term{t}{i}\phi'$.
\end{lemma}
\begin{proof}[Quick sketch]
	By induction on the proof of $\phi$: easy by AN if $\phi$ is an axiom and using the application axiom if $\phi$ is the result of 
%	$\psi, \chi$ and 
	modus ponens.
%	, the theorem holds for some $\term{r}{i}\psi, \term{s}{i}\chi$, so it holds for $\phi$ and $t = [r\cdot s]$. 
\qed
\end{proof}

The Internalization Property demonstrates three important points. One is that a theorem's proof 
%of one of the logic's theorems 
can be internalized as a justification for that theorem. Another point is that Modal Logic's Necessitation rule survives in Justification Logic -- in a weakened form as an axiom and in its full form as a property of the logic. The third point is the importance of the assumption that the constant specification is axiomatically appropriate as it is necessary for the lemma's proof.

\subsection{Semantics}

We present Fitting (F-) models for $J=(n,\subset,\ver,F)_{\cs}$. These are  Kripke models with an additional machinery (an admissible evidence function) to accommodate justification terms. They were introduced  by Fitting in \cite{Fit05APAL} with  variations appearing in \cite{Pac05PLS,DBLP:conf/csr/Kuznets08}.

\begin{definition}
	%		Let $J=(n,\subset,\ver,F)_{\cs}$. Then, 
	An F-model $\M$ for $J$ is a quadruple $(W, (R_i)_{i \in N}, (\A_i)_{i \in N},\V)$, where $W \neq \emptyset$ is a set, for every $i\in N$, $R_i\subseteq W^2$ is a binary relation on $W$, $\V:Pvar \To 2^{W}$ and for every $i\in N$, $\A_i:(Tm\times L_n) \To 2^{W}$. $W$ is called the \emph{universe} of $\M$ and its elements are the worlds or states of the model. $\V$ assigns a subset of $W$ to each propositional variable, $p$, and $\A_i$ assigns a subset of $W$ to each pair of a justification term and a formula. $(\E_{i})_{i\in N}$ is often seen and referred to as $\A : N\times Tm \times L_n \To 2^{W}$ and $\E$ is called an admissible evidence function (aef). 
	Additionally, for any $i\in N$, formulas $\phi, \psi$, and justification terms $t, s$, $\A$ and $(R_i)_{i\in N}$ must satisfy the following conditions:
	\begin{description}
		\item {Application closure:} 
%		for any $i\in N$, formulas $\phi, \psi$, and justification terms $t, s$, \\ 
%		\hfill 
		$ \A_i(s,\phi \rightarrow \psi) \cap \A_i(t,\phi) \subseteq \A_i(s\cdot t, \psi). $
		\item {Sum closure:} 
%		for any $i\in N$, formula $\phi$, and justification terms $t, s$, \\ 
%		\hfill 
		$ \A_i(t,\phi) \cup \A_i(s,\phi) \subseteq \A_i(t+s,\phi).$
		\item {AN-closure:} for any 
		instance of AN, $\term{t}{i}\phi$, $\A_i(t,\phi) = W$.
		\item {Verification Closure:} If $i \ver j$, then $\A_j(t,\phi) \subseteq \A_i(!t,\term{t}{i}\phi)$
		\item {Conversion Closure:} If $i \subset j$, then $\A_j(t,\phi) \subseteq \A_i(t,\phi)$ 
		\item {Distribution:} 
%		for any formula $\phi$, justification term $t$, 
		for $j \ver i$ and $a,b \in W$, if $a R_j b$ and $a \in \A_i(t,\phi)$, then $b \in \A_i(t,\phi)$.\footnote{If we have $\M,a \models \term{t}{i}\phi$ -- and thus $a \in \A_i(t,\phi)$ -- we also want $\M,a \models \term{!t}{j}\term{t}{i}\phi$ to happen and therefore also $\M,b \models \term{t}{i}\phi$ -- so $b \in \E_i(t,\phi)$ must be the case as well.}
	\end{description}
	\begin{itemize}
		\item If $F(i)=\jt$, then $R_i$ must be reflexive.
		\item If $F(i)=\jd$, then $R_i$ must be serial ($\forall a \in W \ \exists b \in W \ a R_i b$).
		\item If $i \ver j$, then for any $a,b,c\in W$, if $a R_i b R_j c$, we also have $a R_j c$.\footnote{Thus, if $i$ has positive introspection (i.e. $i \ver i$), then $R_i$ is transitive.}
		\item For any $i \subset j$, $R_i \subseteq R_j$.
	\end{itemize}
	Truth in the model is defined in the following way, given a state $a$:
	\begin{itemize}
		\item $\M,a \not \models \bot$ and if $p$ is a propositional variable, then $\M,a \models p$ iff $a \in \V(p)$.
		\item 
		$\M,a \models \phi \rightarrow \psi$ if and only if $\M,a \models \psi$, or $\M,a \not \models \phi$.
		\item 
		$\M,a \models \term{t}{i}\phi$ if and only if $a \in \A_i(t,\phi)$ and $\M,b\models \phi$ for all 
		$a R_i b$.
	\end{itemize}
\end{definition}

A formula $\phi$ is called satisfiable if there are  $\M,a \models \phi$; we then say that $\M$ satisfies $\phi$ in $a$.		
A pair $(W,(R_i)_{i\in N})$ as above is a frame for $(n,\subset,\ver,F)_{\cs}$. We say that $\M$ has the 
\emph{Strong Evidence Property} when $\M,a \models \term{t}{i}\phi$ iff $a \in \A_i(t,\phi)$.
%
%\begin{proposition} %[Completeness]
%	\label{completeness}
		$J$ is sound and complete with respect to its F-models;\footnote{That $\cs$ is axiomatically appropriate
		is a requirement for completeness.} it is also complete with respect to F-models with the 
		Strong Evidence property. Furthermore, $J$ has a ``small'' model property, as Proposition \ref{cor:exponentialstates} demonstrates. Completeness is proven in \cite{Achilleos2014CLIMA,Achilleos2014EUMAS} by a canonical model construction with maximally consistent sets of formulas as states; Proposition \ref{cor:exponentialstates} is then proven by a modification of that canonical model construction that depends on the particular satisfiable formula $\phi$.

	\begin{proposition}[\cite{Achilleos2014CLIMA,Achilleos2014EUMAS}] \label{cor:exponentialstates}
		If $\phi$ is $J$-satisfiable, then $\phi$ is satisfiable by an F-model for $J$ of at most $2^{|\phi|}$ states which has the strong evidence property.
	\end{proposition}

	\subsection{The $*$-calculus.}
	\label{sec:starcalc}
	
%	We present the $*$-calculus for $(n,\subset,\ver,F)_{\cs}$. 
	The $*$-calculus gives an axiomatization of $rJ=\{ \phi \in rL_n \mid J \vdash \phi \}$, the reflected fragment of $J$. It is an invaluable tool in the study of the complexity of Justification Logic and when we handle 
%	admissible evidence functions 
	aefs
	and formulas in $rL_n$. 
%	This concept and results were adapted to the two-agent setting in \cite{Achilleos2014CLIMA} and here we extend them to the general multi-agent setting. 
%	Although the calculi have significant similarities to the ones of the single-agent justification logics, there are differences, notably that each calculus depends upon a frame.
	A $*$-calculus was  introduced in \cite{NKru06TCS}, but its origins can be found in \cite{Mkr97LFCS}.

	If $t$ is a term, $\phi$ is a formula, and $i\in N$, then $*_i(t,\phi)$ is a $*$-expression.
	Given a frame $\mathcal{F} = (W,(R_i)_{i\in N})$ for $J$, the $*^\F$-calculus for $J$ is the 
	derivation system on $*$-expressions prefixed by states from $W$ 
	($*^\F$-expressions from now on) with the axioms and rules that are shown in Table \ref{tab:starcalc}.
	
	\begin{table}
		\begin{tabular}{|c|c|}\hline
		\begin{minipage}{0.5\linewidth}
		\begin{description}
		\item[$*\mathcal{CS}(\mathcal{F})$ Axioms:]  $w \ *_i(t,\phi)$, where $\term{t}{i}\phi$ an instance of AN
		\vspace{2ex}
		\item[$*$App$(\F)$:]
		\[ 
		\inferrule*
		{w \ *_i(s,\phi\rightarrow \psi) \\ w \ *_i(t,\phi)}{w \ *_i(s\cdot t , \psi) }
		\] 
		\item[$*$Sum$(\mathcal{F})$:] \[ 
		\inferrule*
		{w \ *_i(t,\phi) }{w \  *_i(s+t,\phi) } \qquad
		\inferrule*
		{w \ *_i(s,\phi) }{w \  *_i(s+t,\phi) }
		\] 
		\end{description}
		\end{minipage}
		&
		\begin{minipage}{0.5\linewidth}
		\begin{description}	
		%			\vspace{2ex}
		\item[$* \ver (\mathcal{F})$:] For any $i \inver j$, \[ 
		\inferrule*
		{w \ *_i(t,\phi) }{w \  *_j(! t,\term{t}{i}\phi) } 
		\] 
		\item[$*\subset(\mathcal{F})$:] For any $i \supset j$, \[ 
		\inferrule*
		{w \ *_i(t,\phi) }{w \  *_j(t,\phi) } 
		\] 
		\item[$* \ver $Dis$(\mathcal{F})$:] For 
		any 
		$i \inver j$,
%		 and 
		 $(a,b)\in R_j$, \[ 
		\inferrule*
		{a \ *_i(t,\phi) }{b \  *_i(t,\phi) } 
		\] 	
		%			\vspace{1ex}
		\end{description}
		\end{minipage}
		\\ \hline
		\end{tabular}
		\label{tab:starcalc}
		\caption{The $*^\F$-calculus for $J$: where $\mathcal{F} = (W,(R_i)_{i\in N})$ and for every $i \in N$}\label{fig:starcalc}
	\end{table}
	For $\Phi \subseteq rL_n$, the $*$-calculus (without a frame) for $J$ can be defined as $\Phi \vdash_* e$ if for every frame $\F$, state $w$ of $\F$, $\{w\ e \mid e\in *\Phi\} \vdash_{*^\F} w\ *_i(t,\phi)$. Notice that for any $v, w$, if $\{w\ e \mid e\in *\Phi\} \vdash_{*^\F} v\ *_i(t,\phi)$, then $\{w\ e \mid e\in *\Phi\} \vdash_{*^\F} w\ *_i(t,\phi)$, therefore the $*$-calculus is the resulting calculus on $*$-expressions after we ignore the frame and world-prefixes (and thus rule $*\ver$Dis($\F$)) in Table \ref{tab:starcalc}. 
%	Finally, f
	For an 
%	admissible evidence function 
	aef
	$\E$, we write $\E \models w\ *_i(t,\phi)$ when $w \in \E_i(t,\phi)$; for set $\Phi$ of $*^\F$- (or $*$-)expressions, $\E \models \Phi$ when $\E \models e$ for every $e \in \Phi$. If $\E \models e$, we may say that $\E$ satisfies $e$.
	\begin{proposition} [\cite{Achilleos2014EUMAS}, but originally \cite{NKru06TCS,Kuz08PhD}]
		\label{prp:proofbystarcalc}
		\begin{enumerate}
		\item
		Let $\Phi \subseteq rL_n$.
%		$\F$ be a frame and $w$ a state. 
		Then, 
		$*\Phi \vdash_{*} e$ iff
%		\begin{enumerate}
%			\item $\Phi \vdash_J \term{t}{i}\phi$
%			\item F
%			$w \in \E_i(t,\phi)$ 
%			$\E \models w\ e$ 
			for any 
%			admissible evidence function 
			aef
			$\E \models *\Phi$, 
			$\E \models w\ e$ .
%			, 
%%			such that 
%			s.t.
%			for every $f \in *\Phi$, $\E \models w\ f$.
			\item For  frame $\F$, set of $*^\F$-expressions $\Phi$, $\Phi \vdash_{*^\F} e$ iff $\E \models e$ for every 
%			admissible evidence function 
			aef
			$\E
%			$ such that $\E 
			\models \Phi$.
		\end{enumerate}
	\end{proposition}
	\begin{proof}
		For 2, notice that the calculus rules correspond to the closure conditions of the
%		admissible evidence function
		aef, 
		so if $\E_m \models e$\footnote{$\E \models e$ has only been defined for 
%		admissible evidence functions
		aefs, 
		but we slightly abuse the notation for convenience.} iff $\Phi \vdash_{*^\F} e$, then $\E_m$ is an 
%	admissible evidence function
	aef, so the ``if'' direction is established; by induction on the calculus derivation, we can also establish for every 
%	admissible evidence function 
	aef
	$\E$, if $\E_m \models e$, then $\E \models e$. 1 is a direct consequence.
		\qed
	\end{proof}
	
%	The proof of the following Proposition \ref{thm:calccompnew} is very similar to the one that can be found in \cite{Kuz08PhD} and can also be found in the Appendix.
	\begin{proposition}[\cite{Achilleos2014EUMAS}, but originally \cite{NKru06TCS,Kuz08PhD}]
		\label{thm:calccompnew}
		If  $\cs \in \P$ and is schematic, the following problems are in $\NP$:
		\begin{enumerate}
			\item
		Given a finite frame $\mathcal{F}$,
%		$ = (W, (R_i)_{i\in N})$, 
		a finite set $S \cup \{e\}$ of $*^\F$-expressions, 
%		prefixed by worlds from $W$, 
%		and
%		a $*^\F$-expression $e$,
%		formula $\term{t}{i}\phi$, 
%		and a $w \in W$, 
		is it the case that $ S \vdash_{*^\F} e
%		w \ *_i(t,\phi)
		\mbox{?} $
		\item Given 
		a finite set $S \cup \{e\}$ of $*$-expressions, 
		%		prefixed by worlds from $W$, 
%		and
%		a $*$-expression $e$,
		%		formula $\term{t}{i}\phi$, 
		%		and a $w \in W$, 
		is it the case that $ S \vdash_{*} e
		%		w \ *_i(t,\phi)
		\text{?} $
	\end{enumerate}
	\end{proposition}
	The shape of a $*$-calculus derivation is mostly described by  $t$. We can use $t$ to extract the general shape of the derivation -- the term keeps track of the applications of all rules besides $*\subset$ and $* \ver $Dis. We can then plug in to the leaves of the derivation either axioms of the calculus or members of $S$ and unify ($\cs$ is schematic, so the derivation includes schemes) trying to reach the root. Using Propositions \ref{thm:calccompnew} and  \ref{cor:exponentialstates}, we can conclude with Corollary \ref{cor:NPandNEXPandP}.

	\begin{corollary}[\cite{Achilleos2014EUMAS}, but 1 was  originally proven in \cite{NKru06TCS}]\label{cor:NPandNEXPandP}
		%Let $J = (n,\subset,\ver,F)_{\cs}$, where $\cs \in \P$ is schematic. Then,
		\begin{enumerate}
			\item If  $\cs \in \P$ and is schematic, then deciding for $\term{t}{i}\phi$ that $J \vdash \term{t}{i}\phi$ is in \NP.
			\item 
			%				If $\cs$ is axiomatically appropriate with respect to $D$, then t
			If  $\cs \in \P$ and is schematic and axiomatically appropriate, then the satisfiability problem for $J$ is in \NEXP.
		\end{enumerate}
	\end{corollary}

\section{A Universal Lower Bound}
\label{sec:universal}

The main result of this section can be found in Theorem \ref{thm:S2lower} and is a lower bound for the complexity of $J$-satisfiability, for an arbitrary multiagent justification logic $J$, given an axiomatically appropriate, schematic constant specification. We give the theorem first and then its proof.

\begin{theorem}\label{thm:S2lower}
If $J$ has an axiomatically appropriate and schematic constant specification, then $J$-satisfiability is $\Sigma^p_2$-hard.	
\end{theorem}

Kuznets proved in \cite{Kuz00CSL} that, under a schematic constant specification, satisfiability for \j, \jt, \jf, and \lp\ is in $\Sigma^p_2$ -- an upper bound which was also successfully established later for \jd\ \cite{Kuz09LC} and \jdf\ \cite{Achilleos2014JCSS} under the assumption of a schematic and axiomatically appropriate constant specification. In that regard, the lower bound of Theorem \ref{thm:S2lower} is optimal. Kuznets' algorithm is composed of a tableau procedure which analyzes signed formulas of the form $T\ \phi$, intuitively meaning that $\phi$ is true in the constructed model, and $F\ \phi$,  meaning that $\phi$ is false, with respect to their propositional connectives (and from $T\ \term{t}{i}\phi$ gives $T\ \phi$ in the presence of Factivity). Eventually it produces formulas of the form $T\ p$, $F\ p$, $T\ *(t,\phi)$, and $F\ *(t,\phi)$, where $T\ *(t,\phi)$ means that the aef of the constructed model makes $(t,\phi)$ \true. The tableau process so far takes polynomial time and makes nondeterministic choices to break the propositional connectives and construct a specific branch. Then we need to make sure that there is a model $(\E,\V)$ such that $\E(t,\phi) = \true$ if $T\ *(t,\phi)$ is in the branch, $\E(t,\phi) = \false$ if $F\ *(t,\phi)$ is in the branch, $\V(p) = \true$ if $T\ p$ is in the branch, and $\V(p) = \false$ if $F\ p$ is in the branch. The propositional variable part is easy to check -- just check that not both $T\ p$ and $F\ p$ are in the branch. The 
%admissible evidence 
aef
part is harder to verify, but the branch 
%conditions are harder, but they 
can give a valid aef if and only if from all $*$-expressions $e$, where $T\ e$ is in the branch we cannot deduce some $*$-expression $f$ using the $*$-calculus, where $F\ f$ in the branch. By Proposition \ref{thm:calccompnew}, this can be verified using an \NP-oracle.

The idea behind the reduction we use to prove Theorem \ref{thm:S2lower} is very similar to Milnikel's proof of $\Pi_2^p$-completeness for \jf-provability \cite{Mil07APAL} (which also worked for \j-provability). Both Milnikel's and our reduction are from $QBF_2$.
% -- with certain differences. 
 The main difference has to do with the way each reduction transforms (or not) the $QBF$ formula. 
Milnikel uses the propositional part of the $QBF$ formula as it is and he introduces existential nondeterministic choices on a satisfiability-testing procedure (think of Kuznets' algorithm as described above) using formulas of the form $\sterm{x}p \vee \sterm{y}\neg p$ and universal nondeterministic choices using formulas of the form $\sterm{x}p \wedge \sterm{y}\neg p$ and term $[x + y]$ in the final term, forcing a universal choice between $x$ and $y$ during the $*$-calculus testing. 

This approach works well for \j\ and \jf, but it fails in the presence of the Consistency or Factivity axiom, as $\sterm{x}p \wedge \sterm{y}\neg p$ becomes inconsistent. For the case of \lp, he used a different approach and made use of his assumption of a schematically injective constant specification (i.e. that all constants justify \emph{at most one} scheme) to construct a term $t$ to specify an \emph{intended} proof of a formula of the form $\bigwedge_i (\sterm{x}p \wedge \sterm{y}\neg p) \rightarrow \sterm{s}\psi$ -- which is always provable, since the left part of the implication is inconsistent. 
In this paper we bypass the problem of the inconsistency of $\sterm{x}p \wedge \sterm{y}\neg p$ by 
replacing each propositional formula 
%of the $QBF$ formula 
by two corresponding propositional \emph{variables},
%introducing new formulas 
$[\chi]^\top$ and $[\chi]^\bot$ 
%for every propositional formula $\chi$, 
to correspond to ``$\chi$ is true'' and to ``$\chi$ is false'' respectively. Therefore, we use $\sterm{x}[p]^\top \wedge \sterm{y}[p]^\bot$ instead of $\sterm{x}p \wedge \sterm{y}\neg p$ and we have no inconsistent formulas. As a side-effect we need to use several extra formulas to encode the behavior of the formulas with respect to a truth-assignment -- for instance, $[p]^\top \rightarrow [p \vee q]^\top$ is not a tautology, so we need a formula to assert its truth (see the definitions of $Eval_j$ below).

Buss and Kuznets in \cite{newlower} use the same assumption as Milnikel on the constant specification to give a general lower bound by a reduction from Vertex Cover and a $\Sigma_2^p$-complete generalization of that problem. Their construction has the advantage that it additionally proves an \NP-hardness result for the reflected fragment of the logics they study, while ours does not. 
On the other hand we do not require a schematically injective constant specification, as, much like Milnikel's construction for \jf, we do not need to limit a $*$-calculus derivation.
% to any intended proofs.

Lemma \ref{lem:universalchoicefor1variable} is a simple observation on the resources (number of assumptions) used by a $*$-calculus derivation: if there is a derivation of $*_i(t,\phi)$ and $t$ only has one appearance of term $s$, then the derivation uses at most one premise of the form $*_j(s,\psi)$. In fact, this observation can be generalized to $k$ appearances of $s$ using at most $k$ premises, but this is not  important for the proof of Theorem \ref{thm:S2lower}.

\begin{lemma}\label{lem:universalchoicefor1variable}
	Let $i$ be an agent, $\phi$ a justification formula, $t$ a justification term in which $!$ does not appear, and $s$ a subterm of $t$ which appears at most once in $t$. Let $S_s = \{ \manyk{\term{s}{i}\phi} \}$ and $S \subset rL_n$, such that $S\cup S_s$ is consistent. 
	Then, $S \cup S_s \vdash \term{t}{i}\phi$ if and only if there is some $1 \leq a \leq k$ such that $S \cup \{\term{s}{i}\phi_a \} \vdash \term{t}{i}\phi$.
\end{lemma}
\begin{proof}
	Easy, by induction on the $*$-calculus derivation (on $t$). \qed
\end{proof}

The proof of Theorem \ref{thm:S2lower} is by reduction from $QBF_2$, which is the following ($\Sigma_2^p$-complete) problem: given a Quantified Boolean Formula, $$\phi = \exists x_1 \exists x_2 \cdots \exists x_k \forall y_1 \forall y_2 \cdots \forall y_{k'} \psi,$$ where $\psi$ is a propositional formula on variables $\manyk{x},\many{y}{k'}$, is $\phi$ true? That is, are there truth-values for $\manyk{x}$, such that for all truth-values for $\many{y}{k'}$, a truth-assignment that gives these values makes $\psi$ true?

As mentioned above, 
for every $\psi_a \in \Psi$, let $[\psi_a]^{\top}, [\psi_a]^{\bot}$ be new propositional variables.
%For the remaining of this section,
%for every propositional formula $\psi$ we introduce propositional variables $[\psi]^\top$ and $[\psi]^\bot$.
%
%The main part of this reduction is the construction from a propositional formula $\phi$ of a term $T^J(\phi)$ and a set $S(\phi) \subset rL_n$.
%, such that for a valuation $v$, 
%$$ \bigwedge_{ v(p_a) = \true } \term{x_a}{i}[p_a]^\top  \wedge \bigwedge_{ v(p_a) = \false } \term{x_a}{i}[p_a]^\bot \wedge S(\phi) \vdash \term{T^J(\phi)}{i} [\phi]^\top$$
%if and only if $v$ makes $\phi$ true.
% That is, we want to internalize propositional truth under a valuation. 
 As we argued earlier, 
% since we are not using the propositional part of the $QBF$ formula or any of its subformulas as is, but we use propositional variables $[\psi]^\bot$ and $[\psi]^\top$ for each such subformula, 
 we need  
formulas to help us evaluate the truth of variables under a certain valuation in a way that matches the truth of the original formula, $\psi$ -- $[\psi]^\bot \rightarrow [\neg \psi]^\top$ for instance. These kinds of formulas (prefixed by a corresponding justification term) are gathered into $S(\phi)$. $T^J(\phi)$ is constructed in such a way that under the formulas of $S(\phi)$ and given a valuation $v$ 
%in the form $$ \bigwedge_{ v(p_a) = \true } \term{x_a}{i}[p_a]^\top  \wedge \bigwedge_{ v(p_a) = \false } \term{x_a}{i}[p_a]^\bot,$$
$$ \bigwedge_{ v(p_a) = \true } \term{x_a}{i}[p_a]^\top  \wedge \bigwedge_{ v(p_a) = \false } \term{x_a}{i}[p_a]^\bot \wedge S(\phi) \vdash \term{T^J(\phi)}{i} [\phi]^\top$$
if and only if $v$ makes $\phi$ true. In other words, $T^J(\phi)$ encodes the method we would use to evaluate the truth value of $\phi$.

\paragraph{To construct $T^J(\phi)$, we first need certain justification terms to encode needed operations to manipulate formulas.} We will often need to work on long conjuncts like $(\phi_1 \wedge \cdots \wedge \phi_r)$, which we can view as a string of formulas. Therefore we need operations like projections ($proj_x^r$), appending a formula  ($append$), appending a formula to a hypothesis ($hypappend$), appending the conclusions of two implications ($appendconc$), and so on. We start by providing these terms. 

 We define terms $proj_x^r$ (for $x \leq r$), $append$,  $hypappend$, and $appendconc$, to be such that $$\term{t}{i}(\phi_1 \wedge \phi_2 \wedge \cdots \wedge \phi_r) \vdash \term{[proj^r_x \cdot t]}{i}\phi_x,$$ 
$$\term{t}{i}\phi_1, \term{s}{i}\phi_2 \vdash \term{ 	[append \cdot t\cdot  s]}{i}(\phi_1\wedge \phi_2),$$ 
$$\term{t}{i}(\phi_1\rightarrow \phi_2) \vdash \term{ 	[hypappend \cdot t]}{i}(\phi_1\rightarrow  \phi_1\wedge \phi_2), \text{ and}$$
$$\term{t}{i}(\phi_1\rightarrow \phi_2), \term{s}{i}(\phi_1\rightarrow \phi_3) \vdash \term{ 	[appendconc \cdot t \cdot s]}{i}(\phi_1\rightarrow  \phi_2 \wedge \phi_3),$$
$append$, $hypappend$, and $appendconc$ can simply be any terms such that 
$$\vdash \term{append}{i}(\phi_1\rightarrow (\phi_2 \rightarrow \phi_1\wedge \phi_2)), $$
$$\vdash \term{hypappend}{i}((\phi_1\rightarrow \phi_2) \rightarrow (\phi_1\rightarrow \phi_1\wedge \phi_2)), \text{ and}$$ 
$$\vdash \term{appendconc}{i}((\phi_1\rightarrow \phi_2) \rightarrow ((\phi_1\rightarrow \phi_3) \rightarrow (\phi_1\rightarrow \phi_2 \wedge \phi_3))).$$ 
Such terms exist, because they justify propositional tautologies and the constant specification is schematic and axiomatically appropriate (see Lemma \ref{internalization}). To define $proj^r_x$, we need terms $left, right, id, tran$, so that 
$$\vdash \term{left}{i}(\phi_1\wedge \phi_2 \rightarrow \phi_1),
%$$
\qquad\qquad\qquad
%$$
\vdash \term{right}{i}(\phi_1\wedge \phi_2 \rightarrow \phi_2),$$
$$\vdash \term{id}{i}(\phi_1\rightarrow \phi_1), \text{ and }
$$
 %\qquad
$$
\vdash \term{tran}{i}((\phi_1\rightarrow \phi_2) \rightarrow ((\phi_2 \rightarrow \phi_3) \rightarrow (\phi_1\rightarrow \phi_3)).$$ Again, such terms exist, because they justify propositional tautologies.
Then, $proj^1_1 = id$; for $r>1$, $proj^r_r = right$; and for $l< r$, $proj^{r+1}_{l} =  [trans \cdot left \cdot proj^r_l]$.

\paragraph{Now we provide the formulas that will help us with evaluating the truth of the propositional part of the $QBF$ formula under a valuation.} These were axioms provided by the constant specification in Milnikel's proof \cite{Mil07APAL}, but as we argued before, we need the following formulas in our case. 
Let $\Psi = \{\many{\psi}{l}\}$ be an ordering of all subformulas of $\psi$, such that if $a < b$, then $|\psi_a| \leq |\psi_b|$\footnote{assume a $|\cdot |$, such that $|p_j|=1$ and if $\gamma$ is a proper subformula of $\delta$, then $|\gamma|<|\delta|$}. 
%For every $\psi_a \in \Psi$, let $[\psi_a]^{\top}, [\psi_a]^{\bot}$ be new propositional variables. 
Furthermore, $\rho = |\{\chi\in \Psi \mid \left| \chi \right| = 1\}|$ and for every $1 \leq j \leq  l$,  
\begin{description}
	\item[if $\psi_j = \neg\gamma$,] then 
%	\[ 
	$
	Eval_j = \term{truth_j}{i}([\gamma]^{\top}\rightarrow [\psi_j]^{\bot}) \wedge  \term{truth_j}{i}([\gamma]^{\bot}\rightarrow [\psi_j]^{\top});
	$
%	\]
	\item[if $\psi_j = \gamma \vee \delta$,] then \[Eval_j = \term{truth_j}{i}([\gamma]^{\top} \wedge [\delta]^{\top}\rightarrow [\psi_j]^{\top}) \allowbreak \wedge \term{truth_j}{i}([\gamma]^{\top} \wedge [\delta]^{\bot}\rightarrow [\psi_j]^{\top}) \]\[ \wedge \term{truth_j}{i}([\gamma]^{\bot} \wedge [\delta]^{\top}\rightarrow [\psi_j]^{\top}) \allowbreak \wedge \term{truth_j}{i}([\gamma]^{\bot} \wedge [\delta]^{\bot}\rightarrow [\psi_j]^{\bot}); \]
	\item[if $\psi_j = \gamma \wedge \delta$,] then \[Eval_j = \term{truth_j}{i}([\gamma]^{\top} \wedge [\delta]^{\top}\rightarrow [\psi_j]^{\top}) \allowbreak \wedge \term{truth_j}{i}([\gamma]^{\top} \wedge [\delta]^{\bot}\rightarrow [\psi_j]^{\bot}) \]\[ \wedge \term{truth_j}{i}([\gamma]^{\bot} \wedge [\delta]^{\top}\rightarrow [\psi_j]^{\bot}) \allowbreak \wedge \term{truth_j}{i}([\gamma]^{\bot} \wedge [\delta]^{\bot}\rightarrow [\psi_j]^{\bot}); \]
	\item[if $\psi_j = \gamma \rightarrow \delta$,] then \[Eval_j = \term{truth_j}{i}([\gamma]^{\top} \wedge [\delta]^{\top}\rightarrow [\psi_j]^{\top}) \allowbreak \wedge \term{truth_j}{i}([\gamma]^{\top} \wedge [\delta]^{\bot}\rightarrow [\psi_j]^{\bot}) \]\[ \wedge \term{truth_j}{i}([\gamma]^{\bot} \wedge [\delta]^{\top}\rightarrow [\psi_j]^{\top}) \allowbreak \wedge \term{truth_j}{i}([\gamma]^{\bot} \wedge [\delta]^{\bot}\rightarrow [\psi_j]^{\top}). \]
\end{description}

\paragraph{We now construct term $T^J(\phi)$.} To do this we first construct terms $T^a$, where $1 \leq a \leq l$. Given a valuation $v$ in the form $\term{x_1}{i}[p_1]^{v_1}, \ldots, \term{x_k}{i}[p_k]^{v_k}$, $T^1$ through $T^k$ simply gather these formulas in one large conjunct (or string). Then for $k+1 \leq a \leq l$, $T^a$ evaluates the truth of $\psi_a$, resulting in either $[\psi]^\top$ or $[\psi]^\bot$ and appending the result at the end of the conjunct.

Let $ T^1 = x_1 $ and for every $1< a \leq k$, $T^a = [append \cdot T^{a-1} \cdot x_a ]$.
It is not hard to see that for $\manyk{v} \in \{\top,\bot \}$,
\begin{equation}\label{eq:rho-buildingvariables}
\term{x_1}{i}[p_1]^{v_1}, \ldots, \term{x_k}{i}[p_k]^{v_k} \vdash \term{T^k}{i}([p_1]^{v_1}\wedge \cdots \wedge [p_k]^{v_k}). 
\end{equation}

If $\psi_a = \neg \psi_{b}$, then 
$$T^a =     hypappend \cdot [trans\cdot proj_b^{a-1} \cdot truth_a] \cdot T^{a-1} \text{ and}$$ 
if $\psi_a = \psi_{b} \circ \psi_c$, then 
$$T^a =     hypappend \cdot [trans\cdot [appendconc \cdot proj_b^{a-1} \cdot proj_c^{a-1}] \cdot truth_a] \cdot T^{a-1}.$$

Let $$ S(\phi) = \bigwedge_{\rho < j \leq l} Eval_j $$
and
given a truth valuation $v$, let 
$$S^v(\phi) = \bigwedge_{v(p_j) = \true} \term{x_j}{i}[p_j]^\top \wedge \bigwedge_{v(p_j) = \false}\term{x_j}{i}[p_j]^\bot  \wedge \bigwedge_{\rho < j \leq l} Eval_j .$$
%and $S(\psi) = \{Eval_j \mid \rho < j \leq l \}.$

By induction on $a$, for every truth assignment $v$, 
$$S^v(\phi) \vdash \term{T^a}{i}([\psi_1]^{v_1} \wedge \cdots \wedge [\psi_a]^{v_a}), $$ 
where if  $\psi_b$ is true under $v$, then $v_b = \top$ and $v_b = \bot$ otherwise. 
The cases where $a \leq k$ are easy to see from (\ref{eq:rho-buildingvariables}). For the remaining cases it is enough to demonstrate that \\
if $\psi_a = \neg \psi_{j}$, then 
$S(\phi) \vdash  \term{[trans\cdot proj_j^{a-1} \cdot truth_a \cdot T^{a-1}]}{i}[\psi_a]^{v_a}$ and \\
if $\psi_a = \psi_{b} \circ \psi_c$, then 
$$S(\phi) \vdash  \term{[trans\cdot [appendconc \cdot proj_b^{a-1} \cdot proj_c^{a-1}] \cdot truth_a \cdot T^{a-1}]}{i}[\psi_a]^{v_a},$$ which is not hard to see by the way we designed each term.

Finally, let $T^J(\phi) = [right \cdot T^l]$.
%, and let $$\phi ^J = \term{[right \cdot T^l]}{i}[\phi]^\top .$$
We can now prove Lemma \ref{lem:internalizetarski}:

\begin{lemma}\label{lem:internalizetarski}
	For every $n \in \nat $ and agent $i \in N$, 
%	there is some $T^J : L_P \To Tm$ and some $S: L_P  \To 2^{rL_n}$, such that 
	$T^J(\phi), S(\phi)$ are computable in polynomial time with respect to $|\phi|$.  $\phi$ is true under truth assignment $v$ if and only if 
	$$ \bigwedge_{ v(p_a) = \true } \term{x_a}{i}[p_a]^\top  \wedge \bigwedge_{ v(p_a) = \false } \term{x_a}{i}[p_a]^\bot \wedge S(\phi) \vdash \term{T^J(\phi)}{i} [\phi]^\top.$$
\end{lemma}

\begin{proof}
	From the above construction we can see that if $\phi$ is true under $v$ then $S^v(\phi) \vdash \term{T^J(\phi)}{i} [\phi]^\top$. On the other hand, if $S^v(\phi) \vdash \term{T^J(\phi)}{i} [\phi]^\top$, then $*S^v(\phi) \vdash_* *_i([right \cdot T^l],[\phi]^\top)$, which in turn gives $(S^v(\phi))^{\#_i} \vdash [\phi]^\top$ (the terms do not include the operator $!$ and thus the right side of a $*$-derivation is a derivation in propositional logic). If  $\phi$ is not true under $v$, then let $v'$ be the valuation, such that $v'([\psi]^\top) = \true$ iff $\psi$ is true under $v$ and $v'([\psi]^\bot) = \true$ iff $\psi$ is false under $v$.
	%which is the same as saying that $\phi$ is true under $v$. 
	Then all of $(S^v(\phi))^{\#_i}$ is true under $v'$ and $[\phi]^\top$ is not, therefore$(S^v(\phi))^{\#_i} \not\vdash [\phi]^\top$, so  $S^v(\phi) \not\vdash \term{T^J(\phi)}{i} [\phi]^\top$.
	\qed
\end{proof}

\begin{corollary}
	The QBF formula $\exists \manyk{p} \forall p_{k+1},\ldots, p_{k+l} \phi$ is true if and only if the following formula is $J$-satisfiable:
	$$\bigwedge_{j=1}^k (\term{x_j}{i}[p_j]^\top \vee \term{x_j}{i}[p_j]^\bot) \wedge \bigwedge_{j=k+1}^l (\term{x_{j}}{i}[p_{j}]^\top \wedge \term{x_{j}}{i}[p_{j}]^\bot) \wedge S(\neg\phi) \wedge \neg T^J(\neg\phi) [\neg\phi]^\top .$$ 
%	is satisfiable.
\end{corollary}
\begin{proof}
	If 
	$$\bigwedge_{j=1}^k (\term{x_j}{i}[p_j]^\top \vee \term{x_j}{i}[p_j]^\bot) \wedge \bigwedge_{j=k+1}^l (\term{x_{j}}{i}[p_{j}]^\top \wedge \term{x_{j}}{i}[p_{j}]^\bot) \wedge S(\neg\phi) \wedge \neg T^J(\neg\phi) [\neg\phi]^\top $$
	is not satisfiable, then 
	$$\bigwedge_{j=1}^k (\term{x_j}{i}[p_j]^\top \vee \term{x_j}{i}[p_j]^\bot) \wedge \bigwedge_{j=k+1}^l (\term{x_{j}}{i}[p_{j}]^\top \wedge \term{x_{j}}{i}[p_{j}]^\bot) \wedge S(\neg\phi)  \vdash T^J(\neg\phi) [\neg\phi]^\top,$$
	and then for every choice $c_1:\{1,\ldots, k\}\To \{\top,\bot\}$,
	$$\bigwedge_{j=1}^k (\term{x_j}{i}[p_j]^{c_1(j)}) \wedge \bigwedge_{j=k+1}^l (\term{x_{j}}{i}[p_{j}]^\top \wedge \term{x_{j}}{i}[p_{j}]^\bot) \wedge S(\neg\phi) \vdash T^J(\neg\phi) [\neg\phi]^\top,$$
	and then since every variable from $\many{x}{k+l}$ appears at most once in $T^J$ and $T^J$ does not include $!$, by Lemma \ref{lem:universalchoicefor1variable} there is some choice $c_2:\{1,\ldots, l\}\To \{\top,\bot\}$ such that
	$$\bigwedge_{j=1}^k (\term{x_j}{i}[p_j]^{c_1(j)}) \wedge \bigwedge_{j=k+1}^l (\term{x_{j}}{i}[p_{j}]^{c_2(j)}) \wedge S(\neg\phi) \vdash T^J(\neg\phi) [\neg\phi]^\top.$$
	Therefore, for every assignment of truth-values on $\manyk{p}$ there truth-values for $p_{k+1},\ldots,p_{l+k}$ that make $\phi$ false.
	
	On the other hand, if 
	$$\bigwedge_{j=1}^k (\term{x_j}{i}[p_j]^\top \vee \term{x_j}{i}[p_j]^\bot) \wedge \bigwedge_{j=k+1}^l (\term{x_{j}}{i}[p_{j}]^\top \wedge \term{x_{j}}{i}[p_{j}]^\bot) \wedge S(\neg\phi) \wedge \neg T^J(\neg\phi) [\neg\phi]^\top$$ is satisfiable,
	then there is some choice $c_1:\{1,\ldots, k\}\To \{\top,\bot\}$, such that
	$$\bigwedge_{j=1}^k (\term{x_j}{i}[p_j]^{c_1(j)}) \wedge \bigwedge_{j=k+1}^l (\term{x_{j}}{i}[p_{j}]^\top \wedge \term{x_{j}}{i}[p_{j}]^\bot) \wedge S(\neg\phi) \wedge \neg T^J(\neg\phi) [\neg\phi]^\top$$ is satisfiable,
	and then since every variable from $\many{x}{k+l}$ appears at most once in $T^J$, for every choice $c_2:\{1,\ldots, l\}\To \{\top,\bot\}$,
	$$\bigwedge_{j=1}^k (\term{x_j}{i}[p_j]^{c_1(j)}) \wedge \bigwedge_{j=k+1}^l (\term{x_{j}}{i}[p_{j}]^{c_2(j)}) \wedge S(\neg\phi) \not\vdash T^J(\neg\phi) [\neg\phi]^\top.$$
	Therefore, there is some truth assignment on $\manyk{p}$ such that every truth assignment on $p_{k+1},\ldots,p_{l+k}$ makes $\phi$ true.
	\qed
\end{proof}

Theorem \ref{thm:S2lower} is then a direct corollary of the above.

%\begin{corollary}
%	If $J$ has an axiomatically appropriate and schematic constant specification, then $J$-satisfiability is $\Sigma^p_2$-complete.
%\end{corollary}

\section{A \NEXP-complete Justification Logic}

The justification logic we prove to have a \NEXP-complete satisfiability problem is the 4-agent logic $J_H = (4,\subset,\ver,F)_\cs$, where 
\begin{itemize}
	\item 
$\subset = \{ (3,4) \}$, 
\item
$\inver = \{ (1,2), (2,3), (4,4) \}$, 
\item
$F(1)=F(2) = \j$,  $F(3)=F(4)=\jd$, and 
\item
$\cs$ is any axiomatically appropriate and schematic constant specification.
\end{itemize}

%$J_H$ is a four-agent logic. 
The agents of $J_H$ are based on justification logics \j\ and \jd\ -- and essentially \jdf, as agent 4 has Positive Introspection.
Agent 3 has a significant variety of justifications. Since $1\inver 2 \inver 3$, 3 is aware of the justifications of 2, who in turn is aware of the justifications of 1. Therefore, 3 can simulate the reasoning of 2 who can simulate the reasoning of 1. Additionally, 3 accepts two types of justifications: the ones 3 receives from 4, which come with Positive Introspection and the other ones 3 accepts, which do not. As Theorem \ref{thm:nexphard} demonstrates, 
this complex interaction among agent 3's justifications results in the significant hardness of $J_H$-satisfiability.

If we only focus on agents $3$ and $4$, we have a \PSPACE-complete justification logic \cite{Achilleos2014CLIMA,Achilleos2014EUMAS}. In a tableau procedure which constructs a model for a given formula (like the one in \cite{Achilleos2014EUMAS}), this means that we may have to consider a large number of states. If we could simply explore smaller parts of the model as we can often do for Modal Logic, we could still end up with an (alternating perhaps) polynomial space algorithm. The satisfiability-testing procedures for Justification Logic have another part, though, and that is testing whether certain $*^\F$-expressions can be derived in a frame $\F$ from a certain set of $*^\F$-expressions using the $*$-calculus -- which corresponds to asking whether there is an aef that satisfies certain expressions and not others. By Proposition \ref{thm:calccompnew}, this can be done using a nondeterministic procedure which takes time polynomial with respect to $|\F|$ and to the overall size of the set of $*^\F$-expressions.
Although the complexity of that procedure is not something which increases the overall complexity of satisfiability-testing \cite{Achilleos2014EUMAS}, to run it we must keep the whole frame $\F$ in memory and $\F$ can be large, which requires exponential time and more than polynomial space. Nondeterminism is introduced as we apply the tableau rules, as some require nondeterministic choices. Assuming $\PSPACE \neq \NEXP$, this is a difficulty we cannot overcome.

\begin{theorem}\label{thm:nexphard}
	$J_H$-satisfiability is \NEXP-hard.
\end{theorem}

The reduction we use 
%is in the Appendix and 
is from a subproblem of the \emph{SCH\"{O}NFINKEL-BERNAYS $\SAT$} problem, which we call \emph{BINARY SCH\"{O}NFINKEL-BERNAYS $\SAT$}:
 \begin{quote}
 	Given a first-order formula $\phi$ of the form 
%\[ 
$\exists x_1 \cdots \exists x_k \forall y_1 \cdots \forall y_{k'} \psi, $
%\] 
where $\psi$ contains no quantifiers or function symbols, is $\phi$ satisfiable by a first-order model of exactly two elements?
 \end{quote} 
 The general \emph{SCH\"{O}NFINKEL-BERNAYS $\SAT$} problem does not require that a satisfying model has exactly two elements and is known to be \NEXP-complete \cite{Lewis1980317}; \emph{BINARY SCH\"{O}NFINKEL-BERNAYS $\SAT$}  remains \NEXP-complete.

The reduction for Theorem \ref{thm:nexphard} is essentially an extended version of the reduction we used to prove Theorem \ref{thm:S2lower}. Like then, consider a construction of a satisfying model, only this time it is an F-model with several states and accessibility relations for agents. Another difference is, of course, that now the original formula is from the first-order language. However, in the \emph{BINARY SCH\"{O}NFINKEL-BERNAYS $\SAT$} formulation, each (first-order) variable is quantified over two possible values (the elements of the two-element model), so they are essentially propositional variables. 
%We also have the relation symbols that operate on these variables, which we can see as $r$-ary propositional connectives, where $r$ the arity of the respective relation symbol. 
Since this is satisfiability we must existentially quantify each relation symbol over all $2^{r+1}$ $r$-ary relations. We can encode such a nondeterministic choice by forcing the existence of an exponential number of states, each representing one $r$-tuple $v=\manyr{v}$ of the two possible values $0$ and $1$ (as mentioned above, we can do this using agents 3 and 4) by having $\term{var}{1}[p_a]^{v_a}$ being true and then at each such state enforce the choice between $\term{rel}{1}[R]^\top$ and $\term{rel}{1}[R]^\bot$, meaning that $v \in R$ or $v \notin R$ respectively -- where $R$ an actual relation. In such a state conjunctions of the form $\term{gather}{1}([p_1]^{v_1}\wedge \cdots \wedge [p_r]^{v_r}\wedge [R]^{\triangle})$ (where $\triangle = \top$ or $\bot$) encode this choice. Due to the particular interaction among the agents and the logics they are based on, in the constructed model $\term{gather}{1}([p_1]^{v_1}\wedge \cdots \wedge [p_r]^{v_r}\wedge [R]^{\triangle})$ is true in a state if and only if that state represents $v$ and $\triangle = \top$ iff $v \in R$. Already this $J_H$-model encodes a first-order model. The trick now is to be able to gather in one state all these formulas that encode the relations through the aef closure conditions (i.e. through the $*$-calculus), but making sure that individual conjuncts (i.e. something of the form $\term{var}{1}[p]^\triangle$ or $\term{rel}{1}[R]^\triangle$) cannot be also transfered to that state through the calculus -- in that case we would be able to construct $\term{gather}{1}([p_1]^{v_1}\wedge \cdots \wedge [p_r]^{v_r}\wedge [R]^{\triangle})$ for additional, invalid combinations of $(v,\triangle)$. This is achieved by considering formulas of the form $\term{!gather}{2}\term{gather}{1}([p_1]^{v_1}\wedge \cdots \wedge [p_r]^{v_r}\wedge [R]^{\triangle})$. The constructed model has empty accessibility relations for agents 1 and 2, thus such formulas can move freely through the accessibility relation of agent 3 (since $2 \inver 3$ and because of Distribution), but this is not the case for anything of the form $\term{t}{1}\chi$ (since $1 \not \inver 3,4$).
Using certain additional formulas we can make sure that $\term{!gather}{2}\term{gather}{1}([p_1]^{v_1}\wedge \cdots \wedge [p_r]^{v_r}\wedge [R]^{\triangle}) \rightarrow [R(\manyr{x})]^\triangle$ becomes true if and only if $\manyr{x}$ are interpreted as $\manyr{v}$.
 The remaining of the formulas and methods we use are very similar to the ones we use for Theorem \ref{thm:S2lower}.

By combining Corollary \ref{cor:NPandNEXPandP} and Theorem \ref{thm:nexphard}, we can claim the following:

\begin{corollary}\label{cor:nexpcomp}
	$J_H$-satisfiability is \NEXP-complete.
\end{corollary}

\subsection{Proof of Theorem \ref{thm:nexphard}}

The reduction we use  is from (a variation of) the \emph{SCH\"{O}NFINKEL-BERNAYS $\SAT$} problem: given a first-order formula $\phi$ of the form 
\[ \exists x_1 \cdots \exists x_k \forall y_1 \cdots \forall y_{k'} \psi, \] 
where $\psi$ contains no quantifiers or function symbols, is $\phi$ satisfiable by a first-order model?

SCH\"{O}NFINKEL-BERNAYS $\SAT$ is known to be $\NEXP$-complete (\cite{Lewis1980317}). 
Furthermore, it is not hard to see that if \[ \exists x_1 \cdots \exists x_k \forall y_1 \cdots \forall y_{k'} \psi, \] is satisfiable, then it is satisfiable by a model of at most $k$ elements. For the coming reduction, we instead use for convenience a simplified version of this problem, which we call \emph{BINARY SCH\"{O}NFINKEL-BERNAYS $\SAT$} and is the same problem, only instead we ask if $\exists x_1 \cdots \exists x_k \forall y_1 \cdots \forall y_{k'} \psi$ is satisfiable by a first-order model of exactly two elements.

For the reductions that follow we use the following notation: for a non-negative integer $x\in \N$, let $bin(x) = bin_0(g),\ldots , bin_{\log g}(g)$ be its binary representation. Furthermore, like in Section \ref{sec:universal}, 
for every propositional and first-order formula $\psi$ we introduce propositional variables $[\psi]^\top$ and $[\psi]^\bot$.

\begin{lemma}
	BINARY SCH\"{O}NFINKEL-BERNAYS $\SAT$
	is \NEXP-complete. 
\end{lemma}
\begin{proof}
Let $\phi$ be 
a first-order formula of the form 
\[ \exists x_1 \cdots \exists x_k \forall y_1 \cdots \forall y_{k'} \psi, \] 
where $\psi$ contains no quantifiers or function symbols. Furthermore, we assume that $\psi$ contains no constants.
We can replace each $x_a$ by $\vec{x}_a = x_a^{1},x_a^{2},\ldots,x_a^{\lceil \log k \rceil }$ and each $y_b$ by $\vec{y}_b = y_b^{1},y_b^{2},\ldots,y_b^{\lceil \log k \rceil }$ in the quantifiers and wherever they appear in a relation. Therefore $\exists x_a$ is replaced by $\exists x_a^1 \exists x_a^2 \cdots \exists x_a^{\lceil \log k \rceil }$ ($\exists \vec{x_a}$ for short) and $\forall x_a$ is replaced by $\forall x_a^1 \forall x_a^2 \cdots \forall x_a^{\lceil \log k \rceil }$ ($\forall \vec{y_a}$ for short) and $R(\manyr{z})$ is replaced by $R(\manyr{\vec{z}})$.  Furthermore, every expression $z = z'$ where $z,z'$ are variables, is replaced by $\bigwedge_{1 \leq a \leq \lceil \log k \rceil} z^a = z'^{a}$ ($\vec{z}=\vec{z'}$ for short). The result of all these replacements in $\psi$ is called $\psi'$. The new formula is:
\[\phi' = \exists \vec{x}_1 \cdots \exists \vec{x}_k \forall \vec{y}_1 \cdots \forall \vec{y}_{k'} \left(\bigwedge_{b = 1}^{k'} \bigvee_{a=1}^{k} \vec{x_a} = \vec{y_b} \rightarrow \psi'\right) \]
We can also define a corresponding transformation of first-order models: assume that the universe of model $\M$ for $\phi$ is a set of at most $k$ natural numbers (each of which is at most $k-1$ and an interpretation for some $x_a$); then $\M'$ is the model with $\{0,1\}$ as its universe, where for every relation $R$ (on tuples of naturals) of $\M$ there is some $R'$, which is essentially the same relation, but on the binary representations of the elements of $\M$. That is, $$R' = \{(bin(a_1),\ldots,bin{(a_r)}) \in \{0,1\}^* \mid  (\manyr{a}) \in R \}$$ It is not hard to see that if $\M$ satisfies the original formula, then $\M'$ satisfies the new one: each $\vec{x}_a$ can be interpreted as the binary representation of the interpretation of $x_a$ in $\M$ and notice that the added equality assertions effectively limit the $\vec{y}$'s to range over the interpretations of the $\vec{x}$'s, which are then exactly the image of the elements of $\M$. 

On the other hand, if $\phi'$ is satisfied by a model with $\{0,1\}$ as its universe, then $\phi$ is satisfied by the model which has the $\lceil \log k \rceil$-tuples of $\{0,1\}$ that are the interpretations of $\manyk{\vec{x}}$ as elements and as relations the restrictions of the two-element model's relations on these tuples.
\qed
\end{proof}

Given a first-order formula $\phi$ as above, we construct a justification formula, $\phi^J$, in polynomial time, such that $\phi$ is satisfiable by a two-element model if and only if $\phi$ is satisfiable by a $J$-model.
The reader will notice several similarities to the proof of Theorem \ref{thm:S2lower}.

Let
\[\phi = \exists x_1 \cdots \exists x_k \forall y_1 \cdots \forall y_{k'} \psi \] 
be such a formula, where $\psi$ contains no quantifiers or function symbols. 
Let $\manym{R}$ be the relation symbols appearing in $\psi$, $\manym{a}$ their respective arities. Then, let $\alpha = \{ i \in \nat \mid \exists r \leq m \text{ s.t. } i \leq a_r \}$; then,  $|\alpha| = \max \{\manym{a}\}$. 
We also define:
$X = \{\manyk{x} \}$; $Y = \{\many{y}{k'} \}$; $Z = X \cup Y$; $\rho_0 = k + k'$.

For this reduction, in addition to the terms introduced in Section \ref{sec:universal}, we define the following justification terms. If we expect a term to justify a tautological scheme of fixed length, then we can just assume the term exists and has some constant size. Otherwise we construct the term in a way that gives it size polynomial with respect to the formula it (provably) justifies. Again we need certain terms to encode manipulations of long conjunctions (which we can see as strings) and we start with these.

\begin{description}
\item[$addhyp$ is such that] $\vdash \term{addhyp}{1}(\phi \rightarrow (\psi \rightarrow \phi))$;
\item[$replaceleft$  is such that]
$\vdash \term{replaceleft}{1}((\phi \rightarrow \phi') \rightarrow ((\phi \wedge \psi) \rightarrow (\phi' \wedge \psi))),$ while
\item[$replaceright$ is such that]
$\vdash \term{replaceright}{1}((\psi \rightarrow \psi') \rightarrow ((\phi \wedge \psi) \rightarrow (\phi \wedge \psi')))$;
\item[We define $replace^k_{l}$ in the following way:]
$$replace^k_k = replaceright,$$ while 
%	\item[
for $l<k$, 
$$replace^k_{l} = trans \cdot replace^{k-1}_l \cdot replaceleft.$$ 
Then it is not hard to see by induction on $k-l$ that 
$$\vdash \term{replace^k_l}{1} ((\phi_l \rightarrow \phi'_l) \rightarrow ((\phi_1 \wedge  \cdots \wedge \phi_l \wedge \cdots \wedge \phi_k) \rightarrow (\phi_1 \wedge \cdots \wedge \phi'_l \wedge \cdots \wedge \phi_k))).$$
\item[We define $mphypoth$]  to be such that $$\vdash \term{mphypoth}{1} ((\phi \rightarrow \psi) \rightarrow ((\phi \rightarrow (\psi \rightarrow \chi)) \rightarrow (\phi \rightarrow \chi))).$$

\item[We use justification variables $\many{var}{a_r}, rel_r$]  for every $r \in [m]$. 
\item[For $1 \leq r \leq m$ we define $ gather_r $ in the following way: ] 
\[gather_r = [append \cdot [append \cdots [ append \cdot var_1 ] \cdots var_{a_r}] \cdot rel_r], \]
%where for every $1 \leq r \leq m$, $\many{var}{a_r}, rel_r$ are justification variables. 
For every $1 \leq j \ \leq a_r + 1$, let $v_j, v' _j \in \{ \top,\bot \}$.
%; we will also use propositional variables $\many{p}{a_r}\in \{\manyk{x},\many{y}{k'} \}$.
Then, for propositional variables $\many{p}{a_r}$,
\[
\bigwedge_ {j = 1}^{a_r}\term{var_j}{1}[p_j]^{v_j}  \wedge \term{rel_r}{1} [R_r]^{v_{a_r + 1}}  \vdash 
%\]\[
\term{gather_r}{1} ([p_1]^{v'_1} \wedge \cdots \wedge [p_{a_r}]^{v'_{a_r}} \wedge [R_r]^{v'_{a_r+1}})
\]
if and only if for every $1 \leq j \leq a_r + 1$, $v_j = v'_j$ (see the proof of Lemma \ref{lem:internalizetarski}).
In fact it is not hard to see that if
$$\bigwedge_ {j = 1}^{a_r}\term{var_j}{1}[p_j]^{v_j}  \wedge \term{rel_r}{1} [R_r]^{v_{a_r + 1}}  \vdash 
\term{gather_r}{1} \chi,
$$ 
then $\bigwedge_ {j = 1}^{a_r}[p_j]^{v_j}  \wedge  [R_r]^{v_{a_r + 1}} \vdash \chi$:  operator $!$ does not appear in $gather_r$, so the right-hand side of a corresponding $*$-calculus derivation for $*_1(gather_r,\chi)$ is a propositional derivation of $\chi$ from $[p_1]^{v_1}, \ldots, [p_{a_r}]^{v_{a_r}}, [R_r]^{v_{a_r+1}}$ and some propositional tautologies.

To give some intuition, conjunction $\bigwedge_ {j = 1}^{a_r}\term{var_j}{1}[p_j]^{v_j}  \wedge \term{rel_r}{1} [R_r]^{v_{a_r + 1}}$ means that $(\many{v}{a_r}) \in R_r$ in a corresponding first-order model.
\item[We use justification variables $value_z$ and $match(z,p_l)$] for all $z \in Z$, $l \in \alpha$.
For every $z \in X$, we define $V_z = \term{value_z}{1}[z]^\top \vee \term{value_z}{1}[z]^\bot $;
%
%\[\Psi^v =  \left(\bigwedge_{\substack{v = v_j \\ R_i(\many{v}{a_i})\\ \text{ appears in }\psi}} \term{q_{v^{\top}}^{R_i(\many{v}{a_i})}}{6}(p_j^\top \rightarrow p_j') \right)  \vee  \left( \bigwedge_{\substack{v = v_j \\ R_i(\many{v}{a_i})\\ \text{ appears in }\psi}} \term{q_{v^{\bot}}^{R_i(\many{v}{a_i})}}{6}( p_j^\bot \rightarrow p_j')\right) 
%\]
%
for every $z \in Y$, $V_z = \term{value_z}{1}[z]^\top \wedge \term{value_z}{1}[z]^\bot $.

%\[\Psi^v =  \left(\bigwedge_{\substack{v = v_j \\ R_i(\many{v}{a_i})\\ \text{ appears in }\psi}} \term{q_{v^{\top}}^{R_i(\many{v}{a_i})}}{6}(p_j^\top {\rightarrow} p_j') \right) \wedge \left( \bigwedge_{\substack{v = v_j \\ R_i(\many{v}{a_i})\\ \text{ appears in }\psi}} \term{q_{v^{\bot}}^{R_i(\many{v}{a_i})}}{6}( p_j^\bot {\rightarrow} p_j')\right) 
%\]
We also define
$$Match =  \bigwedge_{\substack{l\in \alpha\\ z \in Z \\ \triangle \in \{ \top, \bot \} }} 
%\left(
\term{match(z,p_l)}{1}([z]^\triangle \rightarrow ([p_l]^\triangle \rightarrow ok_l)) 
%\wedge \term{match(z,p_l)}{1}([z]^\bot \rightarrow ( [p_l]^\bot \rightarrow ok_l))\right)
$$
\item[
For every $ R_r(\vec{z})$ which appears in $\psi$ and $0\leq b \leq a_r$, we define $match_b^{R_r(\vec{z})}$] in the following way: $match_0^{R_r(\vec{z})} = addhyp \cdot gather_r$ and if $b>0$ and $z_b = x_l$ or $z_b = y_{l-k}$, 
then $match^{R_r(\vec{z})}_b$ is defined  to be the term
\[
%match^{R_r(\vec{z})}_b = 
[
mphypoth \cdot
match^{R_r(\vec{z})}_{b-1}
\cdot
[tran \cdot [tran \cdot project^{\rho_1}_l \cdot match(z_b,b)  ] \cdot replace^{a_r+1}_b ]
].
%,
\]
%while if $b>0$ and $z_b = y_{l-k}$, 
%%and $v_{k} = z_l$, 
%then 
%\[match^{R_r(\vec{z})}_b = 
%[
%mphypoth \cdot
%match^{R_r(\vec{z})}_{b-1}
%\cdot
%[tran \cdot [tran \cdot project^{\rho_1}_l \cdot match(y_l,b)  ] \cdot replace^{a_r+1}_b ]
%].
%\]

We can see by induction on $b$ that for every $0\leq b \leq a_r$,
\[ Match,\ \term{gather_r}{1}( [p_{1}]^{v'_{1}}\wedge \cdots \wedge [p_{a_r}]^{v'_{a_r}} \wedge [R_r]^{v_{a_r +1}}) \vdash  
\qquad \qquad\qquad\qquad\qquad\qquad
\]
\[ \vdash \term{match^{R_r(\many{z}{a_r})}_b}{1}\left(
(
[x_1]^{v_1^{\phantom{l}}}\wedge \cdots \wedge [x_k]^{v_k}\wedge [y_1]^{v_{k+1}}\wedge \cdots \wedge [y_{k'}]^{v_{k'+k}}
) 
\rightarrow \right.
\]\[\qquad \qquad\qquad
\rightarrow \left.
(ok_1 \wedge \cdots \wedge ok_b \wedge [p_{b+1}]^{v'_{b+1}}\wedge \cdots \wedge [p_{a_r}]^{v'_{a_r}} \wedge [R_r]^{v_{a_r +1}})\right)  \]
if and only if for every $j \in [a_r]$ and $j' \in [k+k']$, if $z_j = x_{j'}$ or $z_j = y_{j'-k}$, then $v'_j=v_{j'}$.

$Match$ and term $match_b^{R_r(\vec{z})}$ are used to confirm that given an assignment $v$ for variables $\manyk{x},\many{y}{k'}$, a tuple $\vec{z} \in Z^{a_r}$, and a tuple $(\many{v'}{a_{r+1}}) \in \{\top,\bot \}^{a_{r+1}}$, that $(v(z_1),\ldots,v(z_{a_r})) = (\many{v'}{a_r})$, since this is a crucial condition to assert that $[R_r(\vec{z})]^{v_{a_{r+1}}}$ must be true (i.e. 
$R_r(\vec{z})$ is true iff $v_{a_{r+1}} = \top$).

\item[$T^!(match^{R_r(\vec{z})}_b) $ is defined in the following way:] $ $ \\
$T^!(match^{R_r(\vec{z})}_0) = c_\cdot \cdot !addhyph \cdot ! gather_r$ and for $b>0$ and $z_b = y_{l-k}$,  \\
$T^!(match^{R_r(\vec{z})}_b) = $
\[
c_\cdot \cdot 
[c_\cdot \cdot 
! mphypoth \cdot
T^!(match^{R_r(\vec{z})}_{b-1})]
\cdot
! [ tran \cdot [tran \cdot project^{\rho_1}_l \cdot match(y_l,b)  ] \cdot replace^{a_r+1}_b ]
\]

We can see by induction on $b$ that for every $0\leq b \leq a_r$,
\[ Match,\ \term{!gather_r}{2}\term{gather_r}{1}( [p_{1}]^{v'_{1}}\wedge \cdots \wedge [p_{a_r}]^{v'_{a_r}} \wedge [R_r]^{v_{a_r +1}}) \vdash  
\qquad \qquad\qquad
\]
\[ \vdash \term{T^!(match^{R_r(\vec{z})}_b)}{2}\term{match^{R_r(\vec{z})}_b}{1}\left(
%(
\bigwedge [x_i]^{v_i}
%\wedge \cdots \wedge [x_k]^{v_k}
\wedge 
\bigwedge [y_i]^{v_{k+i}}
%\wedge \cdots \wedge [y_{k'}]^{v_{k'+k}}
%) 
\rightarrow \right.
\]\[\qquad \qquad
\rightarrow \left.
(ok_1 \wedge \cdots \wedge ok_b \wedge [p_{b+1}]^{v'_{b+1}}\wedge \cdots \wedge [p_{a_r}]^{v'_{a_r}} \wedge [R_r]^{v_{a_r +1}})\right)  \]
if and only if 
\[ Match,\ \term{gather_r}{1}( [p_{1}]^{v'_{1}}\wedge \cdots \wedge [p_{a_r}]^{v'_{a_r}} \wedge [R_r]^{v_{a_r +1}}) \vdash  
\qquad \qquad\qquad
\]
\[ \vdash \term{match^{R_r(\many{z}{a_r})}_b}{1}\left(
\bigwedge [x_i]^{v_i}
%\wedge \cdots \wedge [x_k]^{v_k}
\wedge 
\bigwedge [y_i]^{v_{k+i}}
%\wedge \cdots \wedge [y_{k'}]^{v_{k'+k}}
%) 
\rightarrow \right.
\]\[\qquad \qquad
\rightarrow \left.
(ok_1 \wedge \cdots \wedge ok_b \wedge [p_{b+1}]^{v'_{b+1}}\wedge \cdots \wedge [p_{a_r}]^{v'_{a_r}} \wedge [R_r]^{v_{a_r +1}})\right),  \]
which in turn, as we have seen above, is true
if and only if for every $j \in [a_r]$ and $j' \in [k+k']$, if $z_j = x_{j'}$ or $z_j = y_{j'-k}$, then $v'_j=v_{j'}$.
\end{description}

Using the terms (and formulas) we have defined above, we can construct terms $T^a$, where $0<a\leq \rho_1$ and eventually $t^\phi$:

Let $\Psi = \{\many{\psi}{l}\}$ be an ordering of all subformulas of $\psi$ and of variables $\manyk{x},\many{y}{k'}$, which extends the ordering $\manyk{x},\many{y}{k'}$, such that if $a < b$, then $|\psi_a| \leq |\psi_b|.$\footnote{assume a $|\cdot |$, such that $|x_j| = |y_j|=0$, $|R_j(\many{v}{a_j})| = 1$ and if $\gamma$ is a proper subformula of $\delta$, then $|\gamma|<|\delta|$} 
%For every $\psi_a \in \Psi$, let $[\psi_a]^{\top}, [\psi_a]^{\bot}$ be new propositional variables. 
Furthermore, $\rho_0 = |\{a\in [l] \mid \left|{\psi_a}\right| = 0\}|$ ($=k+k'$) and $\rho_1 = |\{a\in [l] \mid \left|{\psi_a}\right| = 1\}|$.

Let $ T^1 = value_{z_1} $ and for every $1< a \leq \rho_0$, $T^a = [append \cdot T^{a-1} \cdot value_{z_a} ]$.
It is not hard to see that for $\manyk{v} \in \{\top,\bot \}$,
\begin{equation}\label{eq:rho-buildingvariables2}
\term{value_{z_1}}{1}[z_1]^{v_1}, \ldots, \term{value_{z_k}}{1}[z_k]^{v_k} \vdash \term{T^{\rho_0}}{1}([z_1]^{v_1}\wedge \cdots \wedge [z_k]^{v_k}). 
\end{equation}

For every $a \in [l]$, 
\begin{description}
\item[if $\psi_a = R_r(z_1^a,\ldots,z^{a}_{a_r})$,] then 
\[ Eval_a = \term{truth_a}{2}(\term{[match_{a_r}^{\psi_a} \cdot T^{\rho_0}]}{1}({ok}_1 \wedge \cdots \wedge {ok}_{a_r} \wedge [R_r]^{\top})\rightarrow [\psi_a]^{\top}) \wedge \] 
\[\wedge  \term{truth_a}{2}(\term{[match_{a_r}^{\psi_a} \cdot T^{\rho_0}]}{1}({ok}_1 \wedge \cdots \wedge {ok}_{a_r} \wedge [R_r]^{\bot})\rightarrow [\psi_a]^{\bot});\]
\item[if $\psi_a = \neg\gamma$,] then 
\[ Eval_a = \term{truth_a}{2}([\gamma]^{\top}\rightarrow [\psi_a]^{\bot}) \wedge  \term{truth_a}{2}([\gamma]^{\bot}\rightarrow [\psi_a]^{\top});\]
\item[if $\psi_a = \gamma \vee \delta$,] then \[Eval_a = \term{truth_a}{2}([\gamma]^{\top} \wedge [\delta]^{\top}\rightarrow [\psi_a]^{\top}) \allowbreak \wedge \term{truth_a}{2}([\gamma]^{\top} \wedge [\delta]^{\bot}\rightarrow [\psi_a]^{\top}) \]\[ \wedge\ \term{truth_a}{2}([\gamma]^{\bot} \wedge [\delta]^{\top}\rightarrow [\psi_a]^{\top}) \allowbreak \wedge \term{truth_a}{2}([\gamma]^{\bot} \wedge [\delta]^{\bot}\rightarrow [\psi_a]^{\bot}); \]
\item[if $\psi_a = \gamma \wedge \delta$,] then \[Eval_a = \term{truth_a}{2}([\gamma]^{\top} \wedge [\delta]^{\top}\rightarrow [\psi_a]^{\top}) \allowbreak \wedge \term{truth_a}{2}([\gamma]^{\top} \wedge [\delta]^{\bot}\rightarrow [\psi_a]^{\bot}) \]\[ \wedge\ \term{truth_a}{2}([\gamma]^{\bot} \wedge [\delta]^{\top}\rightarrow [\psi_a]^{\bot}) \allowbreak \wedge \term{truth_a}{2}([\gamma]^{\bot} \wedge [\delta]^{\bot}\rightarrow [\psi_a]^{\bot}) ; \]
\item[if $\psi_a = \gamma \rightarrow \delta$,] then \[Eval_a = \term{truth_a}{2}([\gamma]^{\top} \wedge [\delta]^{\top}\rightarrow [\psi_a]^{\top}) \allowbreak \wedge \term{truth_a}{2}([\gamma]^{\top} \wedge [\delta]^{\bot}\rightarrow [\psi_a]^{\bot}) \]\[ \wedge\ \term{truth_a}{2}([\gamma]^{\bot} \wedge [\delta]^{\top}\rightarrow [\psi_a]^{\top}) \allowbreak \wedge \term{truth_a}{1}([\gamma]^{\bot} \wedge [\delta]^{\bot}\rightarrow [\psi_a]^{\top}). \]
\end{description}
Let $Eval = \bigwedge_{a = \rho_0 +1}^l Eval_a$.
%
%Let $ T^1 = value_{z_1} $ and for every $1< a \leq \rho_0$, $T^a = [append \cdot T^{a-1} \cdot value_{z_a} ]$.
%It is not hard to see that for $\manyk{v} \in \{\top,\bot \}$,
%\begin{equation}\label{eq:rho-buildingvariables2}
%\term{value_{z_1}}{1}[z_1]^{v_1}, \ldots, \term{value_{z_k}}{1}[z_k]^{v_k} \vdash \term{T^{\rho_0}}{1}([z_1]^{v_1}\wedge \cdots \wedge [z_k]^{v_k}). 
%\end{equation}

For $\rho_0 < a \leq \rho_1$, we define $gathrel_a$ in the following way: $$gathrel_{\rho_0 +1} = c_\cdot \cdot  T^!(match^{\psi_a}_{a_{r_a}})$$ and for $\rho_0 +1 < a \leq \rho_1$, 
$$gathrel_{\rho_0 +1} = appendconc \cdot gathrel_{a-1} \cdot [c_\cdot \cdot T^!(match^{\psi_a}_{a_{r_a}})]. $$

Then, 
$$T^{\rho_0 +1} = replace^{\rho_1 - \rho_0}_1 \cdot truth_{\rho_0 +1} \cdot [gathrel_{\rho_1} \cdot !T^{\rho_0}]$$ 
and for $\rho_0+1 < a \leq \rho_1$, 
$$T^{a} = replace^{\rho_1 - \rho_0}_a \cdot truth_{\rho_0 +1} \cdot T^{a-1}.$$

%For every $\psi_a = R_r(\many{z}{a_r})$, we define: 
%if $a = \rho_0+1$, then $T_R^a = trans \cdot match^{\psi_a}_{a_r} \cdot Eval_a$ and for $\rho_0+1 < a \leq \rho_1$,
%$
%T_R^a =  
%appendconc \cdot T^{a-1} \cdot 
%[trans \cdot match^{\psi_a}_{a_r} \cdot Eval_a$ and for $\rho_0+1 < a \leq \rho_1]
%.$
%Then $T^{\rho_1} = T^{\rho_1}_R \cdot T^{\rho_0}$.
%	match^{R_r(\many{v}{a_r})}_{a_r}
%	[trans\cdot [appendconc \cdot proj_b^{a-1} \cdot proj_c^{a-1}] \cdot truth_a] 
%

if $\psi_a = \neg \psi_{2}$, then 
$$T^a =     hypappend \cdot [trans\cdot proj_j^{a - \rho_0 - 1} \cdot truth_a] \cdot T^{a-1} \text{ and }$$
if $\psi_a = \psi_{b} \circ \psi_c$, then 
$$T^a =     hypappend \cdot [trans\cdot [appendconc \cdot proj_b^{a - \rho_0 - 1} \cdot proj_c^{a-1}] \cdot truth_a] \cdot T^{a-1}.$$

We then define $t^\phi = [right \cdot T^l]$. 

\begin{lemma}\label{lem:startomodel}
For every $b \in [\rho_1]$, $j \in [a_{r_b}]$, let $\vec{l^b}=(\many{l^b}{a_{r_b}}) \in \{p_j, \neg p_j \}^{a_{r_b}}$ and $v^b \in \{\top,\bot \}$. Assume that
for every $b_1,b_2 \in [\rho_1]$, if $r_{b_1} = r_{b_2}$ and $\vec{l^{b_1}} = \vec{l^{b_2}}$, then it must also be the case that $v^{b_1} = v^{b_2}$. 
Then,\footnote{For convenience and to keep the notation tidy, we identify $\vec{l^b}$ with $l_1^{b}\wedge \cdots \wedge l_{a_{r_{b}}}^{b}$ and $\vec{ok}$ with $ok_1\wedge \cdots \wedge ok_{a_{r_{b}}}$.} 
\[ \bigwedge_{b \in [\rho_1]}\term{!gather_{r_b}}{2}\term{gather_{r_b}}{1} 
\left( \vec{l^{b}} \wedge [R_{r_b}]^b\right)
\wedge  Match \wedge Eval  \wedge \bigwedge_{z \in Z} \term{val_z}{1}[z]^{v_z}
\vdash \term{t^\phi}{2} [\phi]^{\top} \]
if and only if $\M \models \phi$ for every model $\M$ with universe $\{ \top, \bot \}$ and interpretation $\I$ such that
\begin{itemize}
\item 
for every $z \in Z$, $v_z = \I(z)$,
\item 
%	for every $b \in [\rho_1]$ and $j \in [a_{r_b}]$, $\I(z_j) = \top$ iff $  $, and
%	\item 
for every $b \in [\rho_1]$, $\M \models R_{r_b}(f(l_1^b),\ldots,f(l^b_{a_{r_b}}))$ iff $v^b = \top$,
\end{itemize}
where for all $j \in \alpha$,  $ f(p_j) = \top $ and $ f(\neg p_j) = \bot $.
\end{lemma}
\begin{proof}

The \emph{if} direction is not hard to see by (induction on) the construction of the terms $T^a, t^\phi$.
For the other direction, notice that a $*$-calculus derivation for \[ \bigwedge_{b \in [\rho_1]}\term{!gather_{r_b}}{2} \term{gather_{r_b}}{1} \left(\vec{l^{b}}\wedge [R_{r_b}]^{v^b}\right), \qquad \qquad
\]
\[
\qquad \qquad 
Match,\  Eval, \ \bigwedge_{z \in Z} \term{val_z}{1}[z]^{v_z}
\vdash 
%\]\[ \vdash 
\term{t^\phi
%	(gather_{r_1},\ldots,gather_{r_k})
}{2}
[\phi]^{\top}
\]
gives on the right hand side a derivation of 
\[ \bigwedge_{b \in [\rho_1]}\term{gather_{r_b}}{1} \left(\vec{l^{b}}\wedge [R_{r_b}]^{v^b}\right),
%\]
%\[
Match,\  Eval^{\#_2}, \ \bigwedge_{z \in Z} \term{val_z}{1}[z]^{v_z}
\vdash 
%\]\[ \vdash 
[\phi]^{\top}
\]
Some $\chi = [R_r(\vec{z}_r^a)]^{\triangle}$, where $R_r(\vec{z}_r^a) = \psi_a$, a subformula of $\phi$, can be derived from the assumptions above only if $\term{[match^{\psi_a}_{a_{r_a}} \cdot T^{\rho_0}]}{1}(\vec{ok}\wedge [R_{r_a}]^\triangle)$ can be derived as well -- notice that the assumptions cannot be inconsistent and we can easily adjust a model that does not satisfy $\term{[match^{\psi_a}_{a_{r_a}} \cdot T^{\rho_0}]}{1}(\vec{ok}\wedge [R_{r_a}]^\triangle)$ so that it does not satisfy $\chi$ either,  by simply changing the truth value of $\chi$.

The derivation of $\term{match^{\psi_a}_{a_{r_a}}}{1}(\vec{ok}\wedge [R_{r_a}]^\triangle)$ is not affected by $Eval^{\#_2}$: if there is a model that satisfies all assumptions except for $Eval^{\#_2}$ and not $\term{match^{\psi_a}_{a_{r_a}}}{1}(\vec{ok}\wedge [R_{r_a}]^\triangle)$, we can assume the strong evidence property and change the truth-values of every $[\psi_b]^{\triangle'}$ to \true, so the new model satisfies all the assumptions and not $\term{[match^{\psi_a}_{a_{r_a}} \cdot T^{\rho_0}]}{1}(\vec{ok}\wedge [R_{r_a}]^\triangle)$.

Therefore we have a $*$-calculus derivation of $\term{[match^{\psi_a}_{a_{r_a}} \cdot T^{\rho_0}]}{1}(\vec{ok}\wedge [R_{r_a}]^\triangle)$ and since $gather_r$ only appears once in $match^{\psi_a}_{a_{r_a}}$, there is some $b \in [\rho_1]$ such that (see Lemma \ref{lem:universalchoicefor1variable})
\[
\term{gather_{r_b}}{1} \left(\vec{l^{b}}\wedge [R_{r_b}]^{v^b}\right),
Match, \bigwedge_{z \in Z} \term{val_z}{1}[z]^{v_z}
\vdash 
\term{[match^{\psi_a}_{a_{r_a}}\cdot T^{\rho_0}]}{1}(\vec{ok}\wedge [R_{r_a}]^\triangle)
\]
Similarly, we can remove the terms from this derivation, so 
\[
\vec{l^{b}}, [R_{r_b}]^{v^b},
Match^{\#_1},\ \bigwedge_{z \in Z} [z]^{v_z}
\vdash 
%\term{match^{\psi_a}_{a_{r_a}}}{1}(
\vec{ok}\wedge [R_{r_a}]^\triangle
%)
\]
From which it is not hard to see that for all $z \in Z$, $v^b = \triangle$, so every first-order model as described in the Lemma satisfies $\chi$. Then it is not hard to see by induction that all such models satisfy all $[\psi_a]^{\triangle}$ derivable from these same assumptions.
\qed
\end{proof}

Now to construct the actual formula the reduction gives. For this let $\rho$ be a fixed justification variable. We define the following formulas.

\[start = \neg [active] \wedge \term{\rho}{3}\left( [active] \wedge \bigwedge_{a \in [\alpha]}\term{var_a}{1}\neg p_a \right)\]
\[forward_A = \term{\rho}{4}\left( \bigvee_{a \in [\alpha]} \term{var_a}{1}\neg p_a \wedge [active] \rightarrow \term{\rho}{3}[active] \right)  \]
\[forward_B = \term{\rho}{4} \bigwedge_{a \in [\alpha]} \left( \bigwedge_{b \in [a-1]} \term{var_b}{1}p_b \wedge \term{var_a}{1} \neg  p_a \wedge [active] 
\right. \]\[ \left. \qquad\qquad
\rightarrow \term{\rho}{3}\left( 
\bigwedge_{b \in [a-1]} \term{var_b}{1} \neg p_b \wedge \term{var_a}{1} p_a
\right)  \right)\]
\[forward_C = \term{\rho}{4} \bigwedge_{a \in [\alpha]} \left(  \bigvee_{b \in [a-1]} \term{var_b}{1}\neg p_{b} \wedge \term{var_a}{1} \neg  p_a \wedge [active] 
\right. \]\[ \left. \qquad\qquad
\rightarrow \term{\rho}{3} \term{var_a}{1} \neg p_a  \right)\]
\[forward_D = \term{\rho}{4} \bigwedge_{a \in [\alpha]} \left(  \bigvee_{b \in [a-1]} \term{var_b}{1}\neg p_{b} \wedge \term{var_a}{1}  p_a \wedge [active]
\right. \]\[ \left. \qquad\qquad
\rightarrow \term{\rho}{3} \term{var_a}{1} p_a  \right)\]
\[end = \term{\rho}{4} \left( \bigwedge_{a \in \alpha} \term{var}{1} p_a  \wedge [active] \rightarrow \term{\rho}{4} \neg [active]\right) \]
\[ choice_R = \term{\rho}{4}\left( [active] \rightarrow \term{rel_r}{1} [R_r]^\top \vee \term{rel_r}{1} [R_r]^\bot  \right) \]
\[ choice_V = \term{\rho}{4}\left( \neg [active] \rightarrow  \bigwedge_{z \in X}\left(\term{value_z}{1} [z]^\top \vee \term{value_z}{1} [z]^\bot  \right) 
\right. \]\[ \left. \qquad\qquad
\wedge \bigwedge_{z \in Y}\left(\term{value_z}{1} [z]^\top \wedge \term{value_z}{1} [z]^\bot  \right) \right) \]
\[ test = \term{\rho}{4} \left( \neg  [active] \rightarrow Match \wedge Eval \wedge \neg \term{t^{\phi}}{2} [\neg \phi]^T  \right) \]
Then, $\phi^J_{FO}$, the formula constructed by the reduction is the conjunction of these formulas above:
\[start \wedge forward_A \wedge forward_B \wedge forward_C \wedge forward_D \wedge end \wedge choice_R \wedge choice_V \wedge test. \]

\begin{theorem}
$\phi^J_{FO}$ is $J$-satisfiable if and only if $\phi$ is satisfiable by a two-element first-order model.
\end{theorem}

\begin{proof}
%	For the purposes of this proof we identify $0,1$ with $\bot, \top$ respectively. For a non-negative integer $g\in \N$, let $bin(g) = bin_0(g),\ldots , bin_{\log g}(g)$ be its binary representation.
First, assume  $\phi$ is satisfiable by two-element first-order model, say $\M$ with interpretation $\I$, and assume that for every $a \in [k]$, $\I(x_a)$ is such that $\M \models \many{\forall y}{k'} \psi$. We construct a $J$-model for $\phi^J_{FO}$: $$\M_J = (W,R_1,R_2,R_3,R_4,\E,\V) \text{, where: }$$ 
\begin{itemize}
\item
$W = \{ \sigma \in \N  \mid  \sigma+2 \in [2^\alpha + 2] \}$ (i.e. $\sigma \in \{ -1, 0, 1, 2, \ldots 2^\alpha  \}$); 
\item 
$R_1 = R_2 = \emptyset$, $R_3 = \{ (\sigma,\sigma+1)  \mid  \sigma <2^\alpha \}\cup \{(2^\alpha,2^\alpha) \}$, and\\ $R_4 = \{ (\sigma,\sigma')  \mid  \sigma < \sigma' \} \cup \{(2^\alpha,2^\alpha) \}$;
\item 
$\A$ is minimal such that 
\begin{itemize} \item
$\A_3(\rho,\chi) = \A_4(\rho,\chi) = W$ for any formula $\chi$, 
\item
$\A_1(var_a,p_a) = \{ \sigma \in W \mid \sigma+1 \in [2^\alpha] \text{ and }  bin_a(\sigma) = 1  \}$,
\item 
$\A_1(var_a,\neg p_a) = \{ \sigma \in W \mid  \sigma+1 \in [2^\alpha] \text{ and }  bin_a(\sigma) = 0  \}$,
\item 
$\A_1(rel_r,[R_r]^\top) = \{ \sigma \in W \mid  \sigma+1 \in [2^\alpha] $ and \\ \hfill $  \M \models R_r(bin_0(\sigma),\ldots,bin_{a_r}(\sigma))  \}$,
\item 
$\A_1(rel_r,[R_r]^\bot) = \{ \sigma \in W \mid  \sigma+1 \in [2^\alpha] $ and \\ \hfill $ \M \not\models R_r(bin_0(\sigma),\ldots,bin_{a_r}(\sigma))  \}$,
\item for every $a \in [k]$,
$\A_1(value_{x_a},[x_a]^\top) = \{ 2^\alpha \}$, if $\I(x_a) = \top$ and $\emptyset$ otherwise,
\item for every $a \in [k]$,
$\A_1(value_{x_a},[x_a]^\bot) = \{ 2^\alpha \}$, if $\I(x_a) = \bot$ and $\emptyset$ otherwise,
\item for every $a \in [k']$,
$\A_1(value_{y_a},[y_a]^\top) =  \A_i(value_{y_a},[y_a]^\bot) = \{ 2^\alpha \}$, and
\item $\M_J,2^\alpha \models Match, Eval$;
\end{itemize}
\item $\V([active]) = \{ \sigma \in W \mid \sigma+1 \in [2^\alpha] \}$ and for any other propositional variable $q$, $V(q) = \emptyset$.
\end{itemize}
It is not hard  to verify that $\M_J,-1 \models \phi_{FO}^J$, as long as we establish that $\M_J,2^\alpha \not \models \term{t^\phi}{2}[\neg \phi]^T$, for which it is enough that $2^\alpha \notin \A_j(t^\phi,[\neg \phi]^\top )$. 

%Meaning I could just replace the one above...
The definition of $\A$ is equivalent to $\sigma \in \A_g(s,\chi) \Leftrightarrow S \vdash_* \sigma\ *_g(s,\chi)$, where $S = $
\[ \{
w\ *_3(\rho,F)\mid w \in W, F \text{ a formula}\}\cup \{
w\ *_4(\rho,F)\mid w \in W, F \text{ a formula}\}\ \cup \]\[ \{
w\ *_1(var_a,p_a)\mid w+1 \in [2^\alpha] \text{ and }  bin_a(w) = 1\}\ \cup \]\[ \{
w\ *_1(var_a,\neg p_a)\mid w+1 \in [2^\alpha] \text{ and }  bin_a(w) = 0\}\ \cup \]\[\{
w\ *_1(rel_r,[R_r]^\top)\mid  w+1 \in [2^\alpha] \text{ and }  \M \models R_r(bin_0(w),\ldots,bin_{a_r}(w))  \}\ \cup \]\[ \{
w\ *_1(rel_r,[R_r]^\bot)\mid  w+1 \in [2^\alpha] \text{ and }  \M \not \models R_r(bin_0(w),\ldots,bin_{a_r}(w))  \}\ \cup  \]\[ \{
2^\alpha\ *_1(value_{x_a},[x_a]^\top) \mid a\in [k],\ \I(x_a) = \top\}\ \cup\]\[ \{
2^\alpha\ *_1(value_{x_a},[x_a]^\bot) \mid a\in [k],\ \I(x_a) = \bot\} \ \cup  \]\[ \{
2^\alpha\ *_1(value_{y_a},[y_a]^\top) \mid a\in [k']\} \cup \{
2^\alpha\ *_1(value_{y_a},[y_a]^\bot) \mid  a\in [k'] \} \ \cup  \]\[ \{
2^\alpha\ e \mid e \in *Eval \cup *Match
\} \]
Then, $2^\alpha \in \A_2(t^\phi,[\neg \phi]^\top )$ iff $S \vdash_* 2^\alpha *_2(t^\phi,[\neg \phi]^\top )$. 
Notice the following: since $t^\phi$ does not have $\rho$ as a subterm, the $*$-expressions in
\[ \{
w\ *_3(\rho,F)\mid w \in W, F \text{ a formula}\}\cup \{
w\ *_4(\rho,F)\mid w \in W, F \text{ a formula}\} \]
cannot be a part of a derivation for $S \vdash_* 2^\alpha *_2(t^\phi,[\neg \phi]^\top )$. 

Since $1 \inver 2 \inver 3$ and $1,2$ do not interact with any agents in any other way, for any term $s$ with no $!$, if for some $a$ or $r$, $var_a$ or $rel_r$ are subterms os $s$, if $S \vdash_* w\ \term{s}{a}\chi$, then $a=1$, $0\leq w < 2^\alpha$, and $\{ w\ e \in S \} \vdash_*  w\ \term{s}{1}\chi$.
%		, while 
%		if $S \vdash_* w'\ \term{s}{1}\chi$, then there is some $0\leq w < 2^\alpha$ and $\{ w\ e \in S \} \vdash_*  w\ \term{s}{1}\chi$, while 
$t^\phi$ includes exactly one $!gather_{r_b}$ for every $b$ and one of $value_z$ for every $z \in Z$. Therefore, if 
$S \vdash_* 2^\alpha *_2(t^\phi,[\neg \phi]^\top )$, then there are
\[ \bigwedge_{b \in [\rho_1]}\term{!gather_{r_b}}{2}\term{gather_{r_b}}{1} 
\Phi \wedge  Match \wedge Eval  \wedge \bigwedge_{z \in Z} \term{val_z}{1}[z]^{v_z}
\vdash \term{t^\phi}{2} [\neg\phi]^{\top} \]
and by Lemma \ref{lem:startomodel}, $\M \models \neg \phi$,  a contradiction.

On the other hand, let there be some $\M'_J$ where $\phi^J$ is satisfied. Then, we name $-1$ a state where $\M'_J,-1 \models \phi^J$ and let $-1 R_3 0 R_3 1 R_3 \cdots R_3 2^\alpha$. Then, 
\begin{itemize} 
\item
$\A_1(var_a,p_a) \subseteq \{ \sigma \in W \mid \sigma+1 \in [2^\alpha] \text{ and }  bin_a(\sigma) = 1  \}$,
\item 
$\A_1(var_a,\neg p_a) \subseteq \{ \sigma \in W \mid  \sigma+1 \in [2^\alpha] \text{ and }  bin_a(\sigma) = 0  \}$,
\item $\M_J,2^\alpha \models Match, Eval$ and for every $a \in [k'],\\ \M_J,2^\alpha \models \term{value_{y_a}}{1}[y_a]^\top, \term{value_{y_a}}{1}[y_a]^\bot$;
\end{itemize}
as we can see by induction on $\sigma$ - the conditions on $A_1(var_a,p_a),A_1(var_a,\neg p_a)$ as imposed by $forward_B,\ forward_C,\ forward_D$ are positive. Notice here that if for some $0\leq w<2^\alpha -1$, $w \in \bigcap_{a \in \alpha}\A_1(var_a,p_a)$, then we have a contradiction: $w+1 \models \neg [active]$ and if $w$ is minimal for this to happen, then $w \models [active]$, so since there is some $a$ s.t. $w \in \A_1(var_a,\neg p_a)$, $w+1 \models [active]$ (by $forward_A$).

Then, $\{w\mid w+1 \in [2^\alpha] \} \subseteq \A_1(rel_r,[R_r]^\top) \cup \A_1(rel_r,[R_r]^\bot)$ and then we can define a first-order model $\M$ such that:
\begin{itemize}
\item 
$\A_1(rel_r,[R_r]^\top) \subseteq \{ \sigma \in W \mid  \sigma+1 \in [2^\alpha] \text{ and }  \M \models R_r(bin_0(\sigma),\ldots,bin_{a_r}(\sigma))  \}$,
\item 
$\A_1(rel_r,[R_r]^\bot) \subseteq \{ \sigma \in W \mid  \sigma+1 \in [2^\alpha] \text{ and }  \M \not\models R_r(bin_0(\sigma),\ldots,bin_{a_r}(\sigma))  \}$,
\item for every $a \in [k]$,
$\A_1(value_{x_a},[x_a]^\top) \subseteq \{ 2^\alpha \}$, if $\I(x_a) = \top$ and $\emptyset$ otherwise,
\item for every $a \in [k]$,
$\A_1(value_{x_a},[x_a]^\bot) \subseteq \{ 2^\alpha \}$, if $\I(x_a) = \bot$ and $\emptyset$ otherwise,
\end{itemize}
Since it must be the case that $\M_J, 2^\alpha \not \models \term{t^{\phi}}{2}{[\neg \phi]}$, it cannot be the case that 
\[ \bigwedge_{b \in [\rho_1]}\term{!gather_{r_b}}{2}\term{gather_{r_b}}{1} 
\Phi \wedge  Match \wedge Eval  \wedge \bigwedge_{z \in Z} \term{val_z}{1}[z]^{v_z}
\vdash \term{t^\phi}{2} [\neg\phi]^{\top} \]
and since $\M$ satisfies the conditions from Lemma \ref{lem:startomodel}, $\M \not \models \neg \phi$.
\qed
\end{proof}

Theorem \ref{thm:nexphard} is then a direct consequence.

\section{Final Remarks}

We gave two lower bounds for the complexity of the satisfiability problem for Justification Logic. Theorem \ref{thm:S2lower} gives a general lower bound which applies to all logics in the family, while  Theorem \ref{thm:nexphard} gives a lower bound for a specific logic in the family. 
From a technical point of view, the reduction from a fragment of QBF that we used for the first result is a simplification of the reduction from a fragment of First-order Satisfiability that we used for the second result. 

The merits of the general $\Sigma_2^p$-hardness result is that we established an (expected) lower bound for all the logics in the family, which uses fewer assumptions than a previous proof of the same bound (for single-agent logics) by Buss and Kuznets in \cite{newlower}. That is, we require a schematic and axiomatically appropriate constant specification, while the proof in \cite{newlower} requires that 
it
%the constant specification 
is also schematically injective: each constant justifies at most one scheme. It is perhaps a subtle distinction, but it means that for the first time we  established this lower bound for justification logics \j, \jt, \jd, \jdf, and \lp, the versions of these single-agent logics with the total constant specification (i.e. the one where all constants justify all axioms).\footnote{If nothing else, this should simplify some of the notation.} The necessity of these properties of the constant specification for these results and their full effects on the complexity of Justification Logic remain to be seen, but some insightful observations were made in \cite{newlower}.

The \NEXP-hardness result we presented in this paper makes the general \NEXP-upper bound from \cite{Achilleos2014EUMAS} tight, thus answering the open question from there about whether there exists a \NEXP-complete logic or the upper bound can be improved. It also makes $J_H$ the first justification logic with known complexity having a harder satisfiability problem (assuming $\EXP \neq \NEXP$) than its corresponding modal logic. In fact, as Proposition \ref{prp:Mhisinexp}, if $M_H$ is the modal logic which corresponds to $J_H$ (the modal logic with the same frame restrictions as $J_H$), then $M_H$-satisfiability is in \EXP: we can simulate the tableau procedure from Table \ref{tab:MHsat} using an exponential time algorithm --  an alternating polynomial space one actually, where we use nondeterministic existential choices to apply the tableau rules and universal choices to select exactly one prefix $\sigma.(g,i)$ from $\sigma$ to explore.
While Modal Satisfiability has been studied extensively, we are not aware of anyone investigating specifically the complexity of $M_H$-satisfiability, so we provide a brief proof.

		\begin{table}[t]
					\begin{minipage}[c][13ex]{0.3\linewidth}
						\[ 
						\inferrule*{\sigma\ T\ \Diamond_i\psi}{\sigma.(g,i)\ T\ \psi }  
						\]
						where $(g,i)$ is new;
					\end{minipage}
					\hfill
		\begin{minipage}[c][13ex]{0.3\linewidth}
						\[ 
						\inferrule*{\sigma\ F\ \Diamond_i\psi}{\sigma.(g,i)\ F\ \psi }  
						\]			
						where $(g,i)$ has already appeared and $i<4$;
		\end{minipage}
		\hfill
		\begin{minipage}[c][13ex]{0.3\linewidth}
		\[ 
		\inferrule*{\sigma\ T\ \Box_i\psi}{\sigma.(g,i)\ T\ \psi }  
		\]
		where $(g,i)$ has already appeared  and $i<4$;
	\end{minipage}
	\hfill 
		\begin{minipage}[c][13ex]{0.3\linewidth}
			\[ 
			\inferrule*{\sigma\ F\ \Box_i\psi}{\sigma.(g,i)\ F\ \psi }  
			\]
			where $(g,i)$ is new;
		\end{minipage}
		\hfill 
			\begin{minipage}[c][13ex]{0.3\linewidth}
				\[ 
				\inferrule*{\sigma\ T\ \Box_i\psi}{\sigma\ T\ \Diamond_i\psi}  
				\]
				where $i \in \{3,4\}$;
			\end{minipage}
	\hfill 
	\begin{minipage}[c][13ex]{0.3\linewidth}
		\[ 
		\inferrule*{\sigma\ T\ \Box_4\psi}{\sigma\ T\ \Box_3\psi}  
		\]
%		where $i \supset j$;
	\end{minipage}
		\\ 
		\begin{minipage}[c][13ex]{0.3\linewidth}
			\[ 
			\inferrule*{\sigma\ T\ \Box_i\psi}{\sigma\ T\ \Box_j\Box_i\psi}  
			\]
			where $0<i<j<4$;
		\end{minipage}
							\hfill
							\begin{minipage}[c][20ex]{0.3\linewidth}
								\[ 
								\inferrule*{\sigma\ F\ \Diamond_4\psi}{\sigma.(g,4)\ F\ \psi \\\\ \sigma.(g,4) F\ \Diamond_4\psi }  
								\]			
								where $(g,i)$ has already appeared and $i \in \{3,4\}$;
							\end{minipage}
							\hfill
							\begin{minipage}[c][20ex]{0.3\linewidth}
								\[ 
								\inferrule*{\sigma\ T\ \Box_4\psi}{\sigma.(g,i)\ T\ \psi \\\\ \sigma.(g,i)\ T\ \Box_4\psi }  
								\]
								where $(g,i)$ has already appeared  and $i \in \{3,4\}$;
							\end{minipage}
		\caption{Tableau rules for $M_H$. To test $\phi$ for $M_H$-satisfiability, start from a branch which only contains $(0,0)\ T\ \phi$ and keep expanding according to the rules above. A branch with $\sigma\ T\ \psi$ and $\sigma\ F\ \psi$ is propositionally closed. A (possibly infinite) branch which is not propositionally closed, but is closed under the rules is an accepting branch.}
		\label{tab:MHsat}
\end{table}

\begin{proposition}\label{prp:Mhisinexp}
	Let $M_H$ be the four-modalities modal logic associated with the class of frames $(W,R_1,R_2,R_3,R_4)$ where $R_3,R_4$ are serial, $R_3 \subseteq R_4$, and for $(i,j) \in \{(1,2),(2,3),(4,4)\}$, if $aR_jbR_ic$, then $aR_ic$.
	Then, $M_H$-satisfiability is in \EXP.
\end{proposition}
	\begin{proof}[Brief]
		We first prove that the tableau procedure from Table \ref{tab:MHsat} is sound and complete. 
		From an accepting branch for $\phi$ we can construct a model for $\phi$: let $W$ be the set of prefixes that have appeared in the branch; let $a \in \V(p)$ iff $a\ T\ p$ has appeared in the branch, let for $i=1,2,3,4$, $r_i = \{(a,a.(g,i)) \in W\times W \}$, for $i = 1,2$, $R_i=r_i$, $R_3$ is the transitive closure of $r_3$, and $R_4$ is the transitive closure of $r_3 \cup r_4$.
		It is not hard to verify that model $\M = (W,R_1,R_2,R_3,R_4)$ satisfies all necessary conditions and that $\M,(0,0) \models \phi$ -- by inductively proving that if $a\ T\ \psi$ in the branch then $\M,a \models \psi$ and if $a\ F\ \psi$ in the branch then $\M,a \not\models \psi$.
		
		On the other hand, from a model $\M = (W,R_1,R_2,R_3,R_4)$ for $\phi$ we can make appropriate nondeterministic choices to construct an accepting branch for $\phi$.
		We map $(0,0)$ to a state $w^{(0,0)}$ such that $\M,w^{(0,0)} \models \phi$; then, when $\sigma.(g,i)$ appears first, it must be because of a formula of the form $\sigma\ T\ \Diamond_i \psi$ (or $\sigma\ F\ \Box_i \psi$, but it is essentially the same case). If $\M,w^\sigma \models \Diamond_i \psi$, then there must be some state $w^{\sigma} R_i w$, such that $\M \models \psi$ and thus we name $w = w^{\sigma.(g,i)}$.
		It is not hard to see that we can make such choices when applying the rules, so that if $a\ T\ \psi$ in the branch then $\M,w^a \models \psi$ if $a\ F\ \psi$ in the branch then $\M,w^a \not\models \psi$. In fact the rules of Table \ref{tab:MHsat} preserve this condition right away; we just need to make sure that the same thing happens with the propositional rules -- for instance, rule $\frac{\sigma\ T\ \psi \vee \chi}{\sigma\ T\ \psi\ \mid\ \sigma\ \T\ \chi}$ can make an appropriate choice depending on whether $\M,w^\sigma \models \psi$ or $\M,w^\sigma \models \chi$. Thus the constructed branch cannot be propositionally closed.
		
		What remains is to show that this tableau procedure can be simulated by an alternating algorithm which uses polynomial space -- thus $M_H$-satisfiability is in {\sf A}$\PSPACE = \EXP$. This can be done by applying the following method: always keep the formulas prefixed by a certain prefix $\sigma$ in memory (at first $\sigma = (0,0)$). First apply all the tableau rules you can on the formulas prefixed by $\sigma$ -- possibly use existential nondeterministic choices for this. Then, using a universal choice, pick one of the prefixes $\sigma' = \sigma.(g,i)$ that were just constructed and replace the formulas you have in memory by the ones prefixed by $\sigma'$. Repeat these steps until we either have $\sigma\ T\ \psi$ and $\sigma\ F\ \psi$ in memory or we see ``enough'' prefixes. In this case, ``enough'' would mean ``more than $2^{6|\phi|}$'', as $\phi$ has up to $|\phi|$ subformulas, so in a branch there can only be up to $6|\phi|$ formulas prefixed by some fixed $\sigma$ -- thus the algorithm only needs to use $O(|\phi|)$ memory and  if it goes through $6|\phi|+1$ prefixes, then two of these have prefixed exactly the same set of formulas.
		If the algorithm accepts $\phi$, then we can easily reconstruct an accepting branch by just taking the union of the constructed formulas, while if there is an accepting branch, then the algorithm can 
		explore only parts of that branch.
%		easily follow the same nondeterministic choices and never reject.
		\qed
	\end{proof}

These results demonstrate a remarkable variability of the system. Although many logics in the family, including the single-agent justification logics, have a $\Sigma_2^p$-complete satisfiability problem,
%(this paper for hardness and \cite{Kuz00CSL,Kuz09LC,Achilleos2014JCSS,Achilleos2014EUMAS}), 
which is lower than the complexity of satisfiability for corresponding modal logics (assuming $\PH \neq \PSPACE$),
there are logics with \PSPACE-complete, \EXP-complete, and as we demonstrated in this paper, \NEXP-complete satisfiability problems, which in the last case is a higher complexity than the one for the corresponding modal logic (assuming $\EXP \neq \NEXP$). Still, it is important to note that even in this case the reflected fragment of the logic remains in \NP\ and in the absence of $+$, in \P.

\subsubsection*{Acknowledgments}
The author is grateful to Sergei Artemov and to an anonymous reviewer; their suggestions significantly enhanced this paper's readability.

\end{document}